%% file: project-SB-arvix-twocol.tex
\documentclass[%
aps,
prl,
reprint,
superscriptaddress,
nofootinbib,
amsmath,amssymb,
floatfix,
]{revtex4-2}

\usepackage{amsthm,mathtools}

\usepackage{graphicx}
\usepackage{dcolumn}
\usepackage{bm}
\usepackage[dvipsnames]{xcolor}
\usepackage{mathrsfs}
\usepackage{nicefrac}
\usepackage{orcidlink}

\usepackage{empheq} 

\newcommand{\dd}{\mathrm{d}}

\newcommand{\gtilde}{\widetilde\gamma}

\newcommand{\Oh}{\mathcal O}

\newtheorem{theorem}{Theorem}
\newtheorem{lemma}[theorem]{Lemma}
\newcommand{\prlsection}[1]{\medskip\par\noindent\emph{#1}---}
\newcounter{prlendmattersubsection}
\newcommand{\prlendmattersection}[1]{%
  \refstepcounter{section}%
  \setcounter{prlendmattersubsection}{0}%
  \medskip\par\noindent\emph{\thesection.~#1}---\ignorespaces}
\newcommand{\prlendmatterheading}[1]{%
  \refstepcounter{prlendmattersubsection}%
  \medskip\par\noindent
  \emph{(\roman{prlendmattersubsection}) #1}---\ignorespaces}

\begin{document}

\preprint{APS/123-QED}

\title{Homogeneous and Isotropic Linearized Gravity as a Caldeira--Leggett System}

\author{Pelayo V. Calzada\,\orcidlink{0009-0006-3133-9359}}
\affiliation{Departamento de Física, Universidad de Alicante, Campus de San Vicente del Raspeig, E-03690 Alicante, Spain}
\author{Pedro Bargueño\,\orcidlink{0000-0001-5453-9042}}
\affiliation{Departamento de Física, Universidad de Alicante, Campus de San Vicente del Raspeig, E-03690 Alicante, Spain}

\begin{abstract}
We show that the homogeneous and isotropic sector of linearized gravity coupled to {unconstrained} bosonic degrees of freedom takes the Caldeira–Leggett form, { when the environment is expanded around a stable equilibrium on the Minkowski background.} These additional degrees of freedom constitute the model’s harmonic reservoir. Eliminating them yields a generalized Langevin equation for the metric perturbation, with their influence {encoded in an effective mass, an effective potential and} the reservoir’s spectral density. This equivalence also yields a natural setting for reduced quantization of the perturbation and stochastic gravitational dynamics. We analyze the resulting dynamics with Laplace methods, extending results presented in other open-system settings.
\end{abstract}

\maketitle

\prlsection{Introduction}%
In many physical situations, the dynamics of a system of interest is affected by degrees of freedom (DOFs) that lie beyond the scope of the analysis. The theory of open systems characterizes their influence through reduced, effective formulations~\cite{Weiss2012,Breuer:2007juk}. A standard approach within the Hamiltonian formalism is provided by system-plus-reservoir models, where the unresolved DOFs are represented by an interacting reservoir, the Caldeira--Leggett model being the paradigmatic realization~\cite{Ullersma:1966, Zwanzig:1973, Caldeira:1981rx, Caldeira:1982uj}. Given the widespread success of this framework across diverse areas of physics and chemistry, including condensed-matter physics \cite{Hanggi:1990zz,Blatter:1994zz,Lindoy_2023}, quantum optics \cite{Louisell1973,Carmichael1993,Gardiner2004}, or biophysical chemistry \cite{Tanimura2006,Ishizaki2009,Cao2020}, it is compelling to explore this approach in gravitation. There, DOFs absent from the classical description, quantum-gravitational or otherwise, leave an imprint on the geometry, which the reservoir may effectively encode.

Several implementations related to this idea have already appeared in gravitational contexts. A Caldeira--Leggett description of spacetime foam can be found in Refs.~\cite{Garay:1998wk,Garay:1998as,Garay:1999cy}, where its effect on low-energy fields on flat spacetime was modeled through a weakly interacting massless field, equivalently a harmonic reservoir, guided by the gas-of-virtual-black-holes picture from the gravitational path integral~\cite{Hawking:1995ag,Addazi:2016jfq}. In loop quantum gravity, the decoherence of a quantum surface was studied in Ref.~\cite{Feller:2016zuk}, postulating the surrounding spin network and matter DOFs as a bilinearly coupled harmonic environment. In a classical metric setting, the Caldeira--Leggett model has been applied to capture the dynamics of the Schwarzschild spacetime as an open system~\cite{Calzada:2025dc}, again postulating a harmonic reservoir as an effective description of unresolved DOFs.

Within the broader system-plus-reservoir approach, stochastic gravity \cite{HuSinha1995,Hu2004Stochastic,Hu2020Semiclassical} and more general open effective field theories (EFTs) \cite{SalcedoColasPajer2024,SalcedoColasDufnerPajer2026} give a classical description of linearized metric perturbations, where the environment (quantum matter DOFs, or of unspecified nature in EFTs) translates into stochastic behavior via the influence functional. In quantum cosmology, open-system treatments of the mini- and midisuperspace Wheeler--DeWitt equation describe environment-induced decoherence \cite{Kiefer1987,Padmanabhan1989,HuPazSinha1993,Calzetta2012Chaos}, and master equations for the reduced quantum gravitational dynamics have been derived in the linearized regime \cite{HollowoodMcDonald2017,Hollowood2018,ShanderaAgarwalKamal2018,TakedaTanaka2025}. A related reduced-state construction appears in Ref.~\cite{HeMitraZurek2026}.


This Letter develops the system-plus-reservoir approach for homogeneous and isotropic linearized gravity, establishing its equivalence to the Caldeira--Leggett model, where the harmonic reservoir follows from an arbitrary environment expanded around a { stable} equilibrium, and the equations of linearized gravity reduce to a generalized Langevin equation. This equivalence provides a connection between linearized gravity and the broader open-systems literature, which we use to analyze the resulting dynamics and to obtain a natural framework for stochasticity and reduced quantization.

\prlsection{Symmetric perturbations of Minkowski spacetime}%
We consider spatially homogeneous and isotropic linear perturbations of the Minkowski metric, analyzing here the vacuum scenario. In the ADM Hamiltonian formalism the associated expansions for a one-parameter family of dynamical variables are
\begin{equation}\label{eq:adm-perturbations}
\begin{aligned}
\gamma_{ij}
&=\delta_{ij}+\varepsilon\varphi\,\delta_{ij}
+\Oh(\varepsilon^2),\\
\pi^{ij}_g
&=\frac{\varepsilon}{3}\,{\pi_{g,\varphi}}\delta^{ij}
+\Oh(\varepsilon^2),\\
N&=1+\varepsilon\alpha+\Oh(\varepsilon^2),
&N^i&=\Oh(\varepsilon^2),
\end{aligned}
\end{equation}
respectively representing the induced spatial metric, conjugate momentum density, lapse function and shift vector.  

For these perturbations the gravitational action, with Gibbons–Hawking–York boundary term, expands as 
\begin{equation}
S_g[\gamma_{ij},\pi^{ij}_g,N,N^i] =\varepsilon^2S_g^{(2)}[\varphi,{\pi_{g,\varphi}}]+\Oh(\varepsilon^3),
\end{equation}
with 
\begin{equation}\label{eq:particle-action-hamil}
S_g^{(2)}[\varphi,{\pi_{g,\varphi}}]
=V_0\int
\left({\pi_{g,\varphi}}\dot\varphi + \frac{\kappa}{3}\,{\pi_{g,\varphi}}^2\right)\dd t,
\end{equation}
where $\kappa=8\pi$ $(G=c=1)$, and $V_0$ is the volume of the fiducial cell, introduced to regularize the spatial integral~\cite{Ashtekar:2011ni}. Note that $S_g^{(2)}$ formally describes the dynamics of a free, nonrelativistic particle of mass $M=-3 / (2\kappa)$.

\prlsection{The Caldeira--Leggett model}%
This model \cite{Ullersma:1966, Zwanzig:1973, Caldeira:1981rx, Caldeira:1982uj} captures the influence of unresolved DOFs on a particle with Hamiltonian $H_S = p^2/2M + V(q)$ under two standard assumptions: the environment stays close to equilibrium, so an (infinite) set of harmonic oscillators represents it, and the coupling is weak enough to linearize the interaction. The global Hamiltonian reads
\begin{equation}\label{eq:Caldeira}
    H =H_S+\frac{1}{2}\sum_{j}\left({\pi_j}^2+\omega_j^2\left(x_j-\frac{\alpha_j}{\omega_j^2} q\right)^2\right),
\end{equation}
with $\omega_j$, ${\pi_j}$, and $x_j$ the frequency, mass-weighted momentum, and mass-weighted position of the $j$-th oscillator, and $\alpha_j$ the coupling strength at frequency $\omega_j$. Integrating out the reservoir DOFs through a standard procedure \cite{Weiss2012} yields, for $t>0$, the generalized Langevin equation (GLE)
\begin{equation}\label{eq:cl-langevin}
M \ddot{q}+\int_0^t \gamma(t-s) \dot{q}(s) \, \dd s +\frac{\partial V}{\partial q}=\xi(t),
\end{equation}
which describes the dynamics of a particle subject to deterministic friction with memory and a stochastic force. The spectral density
\begin{equation}\label{eq:spectral-density}
    J(\omega) = \frac{\pi}{2} \sum_j \frac{\alpha_j^2}{\omega_j} \delta(\omega - \omega_j)
\end{equation}
determines the damping kernel,
\begin{equation}\label{eq:gamma-J-rel}
\gamma(t) = \frac{2}{\pi} \int_0^{\infty} \frac{J(\omega)}{\omega} \cos(\omega t) \, \dd \omega,
\end{equation}
while the stochastic force takes the form
\begin{multline}\label{xi}
\xi(t)=\sum_j \alpha_j \Big[ x_j(0) \cos \left(\omega_j t\right)\\ 
+\frac{{\pi_j}(0)}{\omega_j} \sin \left(\omega_j t\right)\Big] - \gamma(t) q(0). 
\end{multline}

Treating the interacting oscillators as a heat reservoir at temperature $T$, canonically distributed with respect to the reservoir-plus-interaction Hamiltonian, makes $\xi$ a Gaussian process with $\langle \xi(t)\rangle = 0$ and covariance $\langle \xi(t)\xi(s)\rangle = k_B T\, \gamma(t-s)$~\cite{Weiss2012}, so the noise vanishes identically at zero temperature. For an at most quadratic potential, averaging Eq.~\eqref{eq:cl-langevin} shows that the average $\langle q(t)\rangle$ obeys the deterministic form of Eq.~\eqref{eq:cl-langevin} with vanishing right-hand side, so the environmental influence on the average evolution is entirely captured by the spectral density. 

Genuinely irreversible dynamics require promoting the reservoir modes to a continuum, making $J(\omega)$ a smooth function of frequency~\cite{Weiss2012}. A standard choice is a low-frequency power law, $J(\omega) \propto \omega^p$ with $p>0$, supplemented by a high-frequency cutoff at $\omega_c$: sharp, $J(\omega) = \omega^p \, \mathbf 1_{(0,\omega_c)}(\omega)$, or smooth, as in the algebraic form $J(\omega) = \omega^p \, \nicefrac{\omega_c^2}{\omega^2 + \omega_c^2}$.

\prlsection{Minisuperspace-plus-reservoir model}%
The quadratic action of Eq.~\eqref{eq:particle-action-hamil}  formally describes a free particle. This suggests modeling the perturbations of Minkowski spacetime as an open system following a Caldeira--Leggett formalism, where modified dynamics arise from a harmonic reservoir bilinearly coupled to the geometry through the scale factor. The main result of this Letter is that this structure need not be imposed as an ansatz: the homogeneous and isotropic sector of linearized gravity coupled to arbitrary bosonic DOFs {(hereafter the environment)} takes precisely the Caldeira--Leggett form under two  assumptions. { First, the environment has a homogeneous and isotropic\footnote{For a theory invariant under spatial diffeomorphisms, homogeneity and isotropy follow from stability, since translations and rotations preserve the Minkowski geometry and map the reference configuration into equilibria of equal energy, under which a nondegenerate minimum must be invariant.} reference configuration that, together with Minkowski spacetime, solves the unperturbed theory and is a stable equilibrium at fixed Minkowski geometry. Second, the kinetic coefficient of the metric perturbation remains nonzero in the presence of the environment}. 

\medskip 
We consider that {the theory is local} and admits an ADM canonical formulation\footnote{Up to boundary terms, a local tensor-field theory admits this formulation~\cite{Kuchar:1976dtf}.}, so the global action, up to a boundary term, is
\begin{multline}\label{eq:full-canonical-action}
S_{\rm global}=\int\dd t\int\dd^3x\,\Bigl[
\pi^{ij}\dot\gamma_{ij}+p_I\dot q^I\\
{}-N{\mathcal H}-N^i{\mathcal H}_i
\Bigr],
\end{multline}
where $q^I$ ($I=1,\dots,n$) denote the environmental variables, $p_I$ their conjugate momentum densities, and  $\pi^{ij}$ the momentum density conjugate to the spatial metric { in the coupled theory}.

\medskip
The Hamiltonian and momentum constraint densities are decomposed in gravitational and environmental terms as ${\mathcal H} = \mathcal{H}_g + \mathcal{H}_{\rm env}$,  ${\mathcal H}_i = \mathcal{H}_{i, g} + \mathcal{H}_{i, \rm env}$, respectively. { Here $\mathcal{H}_g$ and $\mathcal{H}_{i,g}$ are the vacuum densities evaluated on the canonical metric variables $(\gamma_{ij},\pi^{ij})$ of the coupled theory, and the environmental terms are defined as the remainder.}

{ We consider spatially homogeneous and isotropic linear perturbations of a background $\Theta_0=(1,0,\delta_{ij},0,\bar q^I,\bar p_I)$ with the same symmetries, which solves the full field equations.\footnote{ By homogeneity and isotropy, the background metric momentum is $\bar\pi^{ij}=\tfrac13\bar\pi_\varphi(t)\,\delta^{ij}$. At the background configuration, Hamilton's equation reads $0=\dot\gamma_{ij}=\partial\mathcal H/\partial\pi^{ij}$. Since $\gamma_{ij}=\delta_{ij}$ and $z^A=\bar z^A$ are constant, its time derivative is $\tfrac13(\partial^2\mathcal H/\partial\pi^{ij}\partial\pi^{kl})\,\delta^{kl}\,\dot{\bar\pi}_\varphi=0$, so $\bar\pi_\varphi$ is constant by the nondegeneracy of the kinetic term. The canonical transformation $\pi^{ij}\mapsto\pi^{ij}-\tfrac13\bar\pi_\varphi\sqrt\gamma\,\gamma^{ij}$ then sets it to zero without loss of generality.} For a one-parameter family of configurations $\Theta_\varepsilon\equiv(N,N^i,\gamma_{ij},\pi^{ij},q^I,p_I)_\varepsilon$, the metric variables are expanded as in Eq.~\eqref{eq:adm-perturbations}, with $\pi^{ij}=\tfrac{\varepsilon}{3}\pi_\varphi\delta^{ij}+\Oh(\varepsilon^2)$, where $\pi_\varphi$ is now the momentum conjugate to $\varphi$ in the coupled theory, and the environmental variables as}
\begin{equation}\label{eq:expansion}
    z^A\equiv(q^I,p_I) = {\bar z^A} + \varepsilon Z^A + \mathcal{O}(\varepsilon^2),
\end{equation}
where $\bar z^A \equiv (\bar q^I,\bar p_I)$ and $Z^A \equiv (Q^I,P_I)$, with $A=1,\dots,2n$. 

\medskip
{The reference configuration $\bar z ^A$ is a stable equilibrium of the environment at fixed Minkowski geometry, that is, holding $N=1$, $N^i=0$, $\gamma_{ij}=\delta_{ij}$ and $\dot\gamma_{ij}=0$ while the environmental variables are varied\footnote{The gravitational momentum then varies with $z^A$, its value being fixed by Hamilton's equation $\dot\gamma_{ij}=0$.}. Denoting by $\mathcal H^{\rm fix}(z)$ the Hamiltonian density ${\mathcal H}$ restricted to these conditions, this means
\begin{equation}\label{eq:min}
    \left.\frac{\partial \mathcal{H}^{\rm fix}}{\partial z^A}\right|_{\bar z^A} = 0, \quad
    \left.\frac{\partial^2 \mathcal{H}^{\rm fix}}{\partial z^B \partial z^C}\right|_{\bar z^A} Z^B Z^C > 0, \, Z \neq 0.
\end{equation}
In particular, $\dot\Theta_0=0$ by Hamilton's equations.\footnote{For $N=1$ and $N^i=0$, Hamilton's equations for $z^A$ involve only $\partial\mathcal H/\partial z^A$. At fixed geometry, $\partial\mathcal H/\partial\pi^{ij}=\dot\gamma_{ij}=0$, so the chain rule gives $\partial\mathcal H^{\rm fix}/\partial z^A=\partial\mathcal H/\partial z^A$, which vanishes at $\bar z^A$ by Eq.~\eqref{eq:min}.}}

{For these perturbations the global action expands as $S_{\rm global}=S_{\rm global}^{(0)}+\varepsilon^2 S_{\rm global}^{(2)}+\Oh(\varepsilon^3)$, with (see the End Matter, Sec.~\ref{app:env-expansion})
\begin{multline}\label{eq:action-2}
    S^{(2)}_{\rm global}
    = S_g^{(2)}[\varphi,\pi_\varphi]\\
    {}+V_0\int \Bigl(
        P_I \dot Q^I - \mathsf H_{\rm env}[\varphi,\pi_\varphi,Q^I,P_I]
        \Bigr) \dd t,
\end{multline}
 where $S_g^{(2)}$ is the vacuum quadratic action of Eq.~\eqref{eq:particle-action-hamil} evaluated on the coupled momentum $\pi_\varphi$. The quadratic environmental Hamiltonian has the general form
\begin{multline}\label{eq:env-ham-2}
    \mathsf H_{\rm env} = \frac12\,\mathsf h_{\varphi\varphi}\,\varphi^2+\varphi\,\mathsf h_{\varphi\pi}\,\pi_\varphi+\frac12\,\mathsf h_{\pi\pi}\,\pi_\varphi^2\\
     + \varphi\,\mathsf h_{\varphi Z}\, Z+ \pi_\varphi\,\mathsf h_{\pi Z}\, Z+ \frac12\, Z^{\mathsf T}\mathsf h_{ZZ}\, Z,
\end{multline}
where $Z\equiv(Z^A)$ and $\mathsf h_{ab}$ are the second derivatives of $\mathcal H_{\rm env}$ at the background with respect to $a,b\in\{\varphi,\pi_\varphi,Z\}$.}

\medskip
{
Under the stability condition of Eq.~\eqref{eq:min}, a canonical transformation preserving $\varphi$ realizes the environment as a harmonic reservoir bilinearly coupled to the metric perturbation (see the End Matter, Sec.~\ref{app:cl-form}), reproducing the Caldeira--Leggett model of Eq.~\eqref{eq:Caldeira}. In these variables the quadratic global action of Eq.~\eqref{eq:action-2} reads
\begin{equation}\label{global-action}
    S^{(2)}_{\rm global} = V_0 \int\bigg(\Pi_\varphi\dot\varphi + \sum_j\pi_j \dot x_j - H_{\rm CL} \bigg) \dd t,
\end{equation}
with
\begin{multline}\label{H_CL}
    H_{\rm CL} = \frac{\Pi_\varphi^2}{2M_{\rm eff}} + V(\varphi)\\+ \frac12
    \sum_j
    \left(
        \pi_j^2
        +
        \omega_j^2
        \left(x_j-\frac{\alpha_j}{\omega_j^2}\,\varphi\right)^2
    \right).
\end{multline}
The metric perturbation thus acquires an effective mass and potential,
\begin{equation}\label{eq:Meff}
    M_{\rm eff}\equiv\left(-\frac{2\kappa}{3}+\mathsf h_{\pi\pi}\right)^{-1},\qquad
    V(\varphi)\equiv\frac{\sigma}{2}\,\varphi^2,
\end{equation}
with $M_{\rm eff}$ well defined by the nondegeneracy assumption, and $\sigma$ determined by the environmental coefficients $\mathsf h_{ab}$ and $M_{\rm eff}$, as given in the End Matter, Eq.~\eqref{eq:sigma-em}. At fixed $\varphi$, $V(\varphi)$ is the stationary value of the quadratic Hamiltonian in Eq.~\eqref{eq:action-2} with respect to $\pi_\varphi$ and $Z$, so $-\sigma\varphi$ is the restoring force on the metric perturbation once the environment has relaxed to its static response.}

\medskip
Following the standard Caldeira-Leggett procedure leads to reduced dynamics for the metric perturbations governed by the GLE of Eq.~\eqref{eq:cl-langevin} with mass $M={M_{\rm eff}}$,
\begin{equation}\label{eq:varphi-langevin}
    {M_{\rm eff}}\,\ddot{\varphi} + \int_0^\eta \gamma(\eta-s)\,\dot{\varphi}(s)\,\dd s + \sigma \varphi = \xi(\eta), \quad \eta>0.
\end{equation}

 As noted below Eq.~\eqref{xi}, the mean perturbation $\langle\varphi\rangle$ satisfies Eq.~\eqref{eq:varphi-langevin} with a vanishing right-hand side. The presence of the fiducial cell volume in the action of Eq.~\eqref{global-action}, although it has no effect on deterministic dynamics, introduces an overall factor of $V_0$ in the reservoir-plus-interaction Hamiltonian, which suppresses fluctuations as $V_0^{-1}$,
\begin{equation}\label{eq:xi-suppressed}
    \langle\xi(\eta)\xi(s)\rangle
    =\frac{k_B T}{V_0}\,\gamma(\eta-s).
\end{equation}

\prlsection{Linearized dynamics}%
{The mean metric perturbation $\langle \varphi \rangle$ obeys Eq.~\eqref{eq:varphi-langevin} with vanishing right-hand side, and its behavior depends first on the sign of the effective mass. For $M_{\rm eff}>0$, Eq.~\eqref{eq:varphi-langevin} is the standard GLE of a damped particle~\cite{Weiss2012}. For $M_{\rm eff}<0$, as in vacuum, $\langle\varphi\rangle$ behaves formally as a unit-mass particle subject to an \emph{anti-damping} force in the quadratic potential $\sigma\varphi^2/(2M_{\rm eff})$. We analyze this case with Laplace methods (see the End Matter, Sec.~\ref{app:laplace}), which classify the evolution according to the sign of $\sigma$ and characterize when Minkowski spacetime coupled to an environment in equilibrium remains stable under these perturbations.}

For $\sigma>0$, the effective potential is repulsive and generic evolution diverges exponentially, so the first-order perturbative approach breaks down.

For $\sigma<0$, the effective potential is harmonic, with natural frequency $\Omega_0=\sqrt{\sigma/{M_{\rm eff}}}$, and the GLE admits nontrivial bounded evolution in the weak-coupling regime of spectral densities with small amplitude. In that case, a necessary condition for bounded evolution is that the spectral density be suppressed around the natural frequency. More precisely, spectral densities bounded between two positive constants near $\Omega_0$ resonantly drive the anti-damping and produce unbounded evolution at arbitrarily weak coupling (see Theorem~\ref{theorem} in the End Matter).

As an illustration of bounded $\varphi$ evolution, for which these conditions are also sufficient, consider the power-law spectral densities with sharp cutoff at $\omega_c<\Omega_0$ \cite{Weiss2012},
\begin{equation}\label{eq:J-sharp}
    J_\zeta(\omega)=\zeta\,\omega^{p}\,\mathbf 1_{(0,\omega_c)}(\omega), \quad \zeta, p>0.
\end{equation}

In the weak-coupling regime, $\zeta\ll1$ (explicitly, $\zeta<\zeta_{\mathrm b}$, as defined in Theorem~\ref{thm:bounded} in the End Matter), the metric perturbation evolves as a bounded oscillation plus a transient term,
\begin{multline}\label{eq:final-thm}
    \varphi_\zeta(\eta)=
   \sum_{k\in\{1,2\}}\bigl[A_k\cos(\nu_k\eta)+B_k\sin(\nu_k\eta)\bigr]\\
    +\varphi_{\zeta,\mathrm{tr}}(\eta),
\end{multline}
with real constants $A_k,B_k$, frequencies $\omega_c<\nu_1<\nu_2<\Omega_0$, and $\varphi_{\zeta,\mathrm{tr}}(\eta)\to0$ as $\eta\to\infty$ (see Theorem~\ref{thm:bounded} in the End Matter).

Dynamics can be determined\footnote{By dimensional analysis, $\varphi(\eta)$ with parameters $(\Omega_0, \omega_c, \zeta)$ and initial values $(\varphi_0, \dot\varphi_0)$ has the same dynamics as $\varphi(\eta/\upsilon)$ with rescaled parameters $(\upsilon \, \Omega_0, \upsilon \, \omega_c, \upsilon^{2-p} \, \zeta)$, and $(\varphi_0, \upsilon \, \dot\varphi_0)$. Choosing $\upsilon=\Omega_0^{-1}$ and using $\zeta_{\mathrm b}=-\pi p{M_{\rm eff}}(\Omega_0-\omega_c)^2/(2\omega_c^p)$ gives $\Omega_0^{p-2}\zeta=-(\zeta/\zeta_{\mathrm b})\,\nicefrac{\pi p{M_{\rm eff}}(1-\omega_c/\Omega_0)^2}{2(\omega_c/\Omega_0)^p}$.} from the parameters $(\omega_c/\Omega_0, \zeta/\zeta_{\mathrm b}, p)$. Fig.~\ref{fig:bounded-ev} displays the transient $\varphi_{\zeta,\mathrm{tr}}$ across the sub-Ohmic ($p<1$), Ohmic ($p=1$), and super-Ohmic ($p>1$) regimes, together with its asymptotic decay laws. For $p<1$ and $\varphi_0\neq0$ the decay is dominated by the contribution of frequencies near $\omega=0$ and follows $\eta^{-p}$. For $p\geq1$, or when $\varphi_0 = 0$, it is dominated by the contribution of frequencies near the cutoff, which oscillates at frequency $\omega_c$ with envelope $\{\eta\,[\ln^{2}(\eta/\eta_*)+\pi^{2}]\}^{-1}$, where $\eta_*$ is a timescale fixed by the spectral density parameters. These decay laws are established in Theorem~\ref{thm:X} in the End Matter, where $\eta_*$ is defined.

{Finally, for $\sigma=0$ the only bounded evolutions are constant and describe Minkowski spacetime up to an overall rescaling of the scale factor.  Spatial diffeomorphism invariance of the coupled action for the metric and the environment implies $\sigma=0$ (see the End Matter, Sec.~\ref{app:covariance-stiffness}). Nonzero values of $\sigma$, and the dynamics described above, can arise from DOFs not retained in the environment, such as fermionic or nonlocal ones, when their influence is encoded in prescribed data, for instance conserved comoving densities fixed by initial data.}

\begin{figure}[t]
     \centering
    \includegraphics[width=\columnwidth]{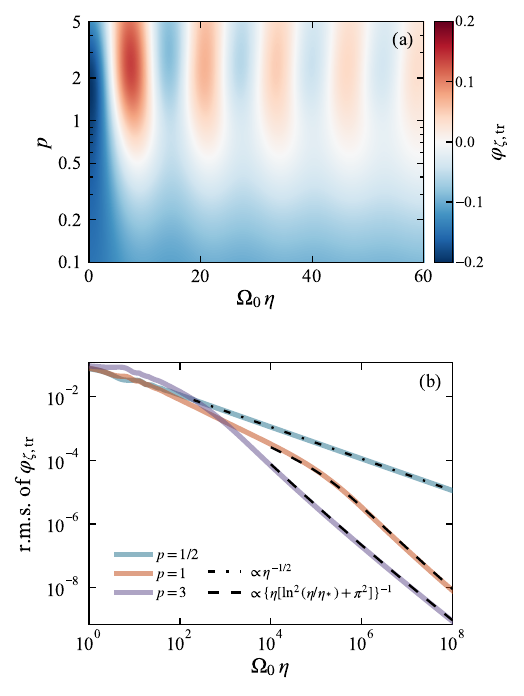} 
    \caption{Transient contribution to the evolution of $\varphi$ for the sharp-cutoff power-law spectral density. (a) \(\varphi_{\zeta,\mathrm{tr}}(\eta)\) for \(0.1\le p\le 5\). (b) root mean square of \(\varphi_{\zeta,\mathrm{tr}}\), computed over the period of the asymptotic behavior \(2\pi/\omega_c\), for sub-Ohmic \((p=1/2)\), Ohmic \((p=1)\), and super-Ohmic \((p=3)\); dashed curves show the analytic asymptotic behavior (see the End Matter). We used reservoirs with parameters $\omega_c/\Omega_0=0.5$ and $\zeta/\zeta_{\mathrm b} = 0.5$, and initial conditions $\varphi_0=1$, $\dot{\varphi}_0=0$.}
\label{fig:bounded-ev}
\end{figure}

\prlsection{Discussion and outlook}%
This Letter establishes the equivalence between the homogeneous and isotropic sector of linearized gravity coupled to arbitrary {unconstrained} bosonic DOFs and the Caldeira--Leggett model, {provided the kinetic coefficient of the metric perturbation remains nonzero}. The harmonic reservoir encoding the environment arises as the second-order expansion, around a homogeneous stable equilibrium, of an arbitrary collection of DOFs. In line with the open-systems approach, this representation is agnostic with respect to the microscopic nature of the environment and depends only on its quadratic Hamiltonian Eq.~\eqref{eq:env-ham-2}. This encompasses  local matter field theories (such as collections of scalar fields or $n$-forms fields \cite{Kuchar:1976dtf}) and also effective DOFs that may capture physics beyond general relativity, including quantum-gravitational corrections. After tracing out the environmental DOFs, the equations of linearized gravity reduce to the GLE of Eq.~\eqref{eq:varphi-langevin}, where the reservoir influence is condensed in the spectral density. {Geometrically, its solutions $\varphi$ encode the gauge-invariant linearized Riemann tensor, whose only nonvanishing components are
\begin{equation}\label{eq:riemann}
    \delta R_{0i0j}=-\frac12\,\ddot\varphi\,\delta_{ij}.
\end{equation}
Since solutions of Eq.~\eqref{eq:varphi-langevin} generically have $\ddot\varphi\neq0$, they are not gauge transformations in general\footnote{In vacuum $\ddot\varphi=0$, and the solutions $\varphi=\varphi_0+vt$ are pure gauge, generated by $\xi^\mu=\bigl(\tfrac{v}{4}\,\delta_{ij}x^ix^j+f(t),\ \tfrac{\varphi}{2}\,x^i\bigr)$, with $\dot f=\alpha$.}.}

The Caldeira--Leggett representation of the gravitational dynamics also yields a reduced quantization of homogeneous and isotropic linearized gravity. Comprehensive treatments of this quantization scheme can be found in Refs.~\cite{Weiss2012,Breuer:2007juk}. The canonical quantization of the quadratic Hamiltonian in Eq.~\eqref{H_CL} proceeds by promoting the perturbation and the reservoir modes to operators satisfying $[\widehat\varphi,V_0 {\widehat\Pi_\varphi}]=i\hbar$ and $[\widehat x_j, V_0 \widehat\pi_k]=i\hbar\,\delta_{jk}$, with the global state evolving unitarily under $V_0 H_{\rm CL}$. The Heisenberg equations for the reservoir operators can be solved, and the reduced dynamics of the geometry is governed by a quantum GLE, the quantum counterpart of Eq.~\eqref{eq:varphi-langevin}
\begin{equation}\label{eq:quantum-langevin}
{M_{\rm eff}}\,\ddot{\widehat\varphi}
+\int_0^\eta\gamma(\eta-s)\dot{\widehat\varphi}(s)\,\dd s
+\sigma\widehat\varphi
=\widehat\xi(\eta),
\end{equation}
where the operator-valued noise $\widehat\xi(\eta)$ is built from the initial reservoir operators as in Eq.~\eqref{xi}. 
For a reservoir in thermal equilibrium, fluctuations persist at zero temperature, and are suppressed as $V_0^{-1}$,
\begin{equation}\label{eq:quantum-noise}
\begin{aligned}
\left\langle\widehat\xi(\eta)\widehat\xi(s)\right\rangle
&=\frac{\hbar}{\pi V_0}\int_0^\infty J(\omega)
\left\{\coth\!\left(\frac{\hbar\omega}{2k_BT}\right)
\cos[\omega(\eta-s)]\right.\\
&\qquad\qquad\left.{}-i\sin[\omega(\eta-s)]\right\}\,\dd\omega.
\end{aligned}
\end{equation}

{This analysis of homogeneous linear perturbations requires regularizing spatial integrals over $\mathbb{R}^3$. Following} a standard approach \cite{Ashtekar:2011ni}{,} we have introduced a fiducial cell and restricted spatial integrals to it. This acts as an infrared regulator which must be removed in order to obtain physical results by taking the limit $V_0 \to \infty$ at the end. In particular, by Eqs.~\eqref{eq:xi-suppressed} and \eqref{eq:quantum-noise}, the stochastic force vanishes (at classical or quantum level). A standard alternative is working with the spatial topology of the 3-torus $\mathbb{T}^3$, for which spatial integrals are well-defined, with $V_0$ now identified with the physical volume of $\mathbb{T}^3$~\cite{Ashtekar:2011ni}. However, this approach introduces a Taub charge constraint~\cite{BrillDeser1973,Higuchi1991Toroidal,AltasTekin2019Taub}, which in our setting takes the form $H_{CL} = 0$. This restriction on the initial data is incompatible with the canonical ensemble distribution associated to the reservoir, and prevents taking the continuum limit.

\medskip
The analysis of the GLE dynamics also benefits from this contact with the open-systems literature. Equation \eqref{eq:varphi-langevin} is, {for $M_{\rm eff}<0$, }a GLE for a particle of negative mass, equivalently a particle subject to an anti-damping force. Its Markovian limit, $\gamma(\eta)\propto\delta(\eta)$, appears in the dynamics of dark solitons \cite{Aycock2017,Hurst2017}. For $\sigma<0$, a negative-energy oscillator immersed in a harmonic continuum has been analyzed within Hamiltonian bifurcation theory \cite{HagstromMorrison2013}, where generic coupling is shown to destabilize the oscillator, in agreement with Theorem~\ref{theorem} in the End Matter. The results presented here extend that analysis to arbitrary memory kernels and give necessary and sufficient conditions for bounded evolution.

Additional bounded evolutions can be identified through the inverse problem of reconstructing the spectral density that generates a prescribed evolution. Reconstruction procedures tailored to noisy or discretized data have been developed across molecular and stochastic modeling~\cite{lange2006collective,daldrop2018butane,satija2019generalized,gottwald2015parametrizing,lei2016data,kerrwinter2023deep,tepper2024accurate}. In the present setting the linearity of Eq.~\eqref{eq:varphi-langevin} suggests a formal answer. The mean evolution equation in Laplace space, Eq.~\eqref{eq:varphi(z)}, solved for the damping kernel yields
\begin{equation}\label{eq:gamma-inverse}
    \widetilde{\gamma}(s) = \frac{-{M_{\rm eff}}(s\varphi_{0}+\dot{\varphi}_{0}) + \left({M_{\rm eff}}s^{2}+\sigma\right)\widetilde{\varphi}(s)}{\varphi_{0} - s\widetilde{\varphi}(s)}.
\end{equation}
Whenever this expression is well defined, the mean evolution determines the damping kernel, and inverting the relation of Eq.~\eqref{eq:gamma-J-rel} yields the associated spectral density, provided the result is a regular non-negative function.

While the analysis illustrated concerns the mean metric perturbation, for which the noise term of Eq.~\eqref{eq:varphi-langevin} vanishes, GLE dynamics also incorporates stochastic effects encoded in the Gaussian process $\xi(\eta)$ fixed by the reservoir.

The equivalence between the homogeneous and isotropic sector of linearized gravity with arbitrary DOFs considered here and the Caldeira--Leggett model motivates future work toward a broader \textit{gravity-plus-reservoir} framework. In this framework, the influence of unresolved microscopic DOFs on the geometry would be encoded in an effective reservoir, providing a distinct connection of gravitational dynamics with open-system methods developed across condensed-matter physics, quantum optics, and chemical physics.

\medskip
\begin{acknowledgments}
PVC acknowledges support from Generalitat Valenciana Grants No. CIACIF/2021/268 and No. CIBEFP/2026/145. PB acknowledges financial support from the Generalitat Valenciana through PROMETEO PROJECT CIPROM/2022/13 and Ministerio de Ciencia, Innovación y Universidades and Agencia Estatal de Investigación (AEI) under Project PID2025-171322NB-C21 (MICIU/AEI/10.13039/501100011033).
\end{acknowledgments}

\bibliography{biblio}

\onecolumngrid
\begin{center}
    \textbf{\large End Matter}
\end{center}
\twocolumngrid

\appendix



\prlendmattersection{{Quadratic action}}\label{app:env-expansion}%
{We expand $S_{\rm global}$ along the family $\Theta_\varepsilon$ of Eqs.~\eqref{eq:adm-perturbations} and \eqref{eq:expansion}, with canonical variables $\Psi_\varepsilon\equiv(\gamma_{ij},\pi^{ij},q^I,p_I)_\varepsilon$, and denote by $F^{(n)}\equiv\frac{1}{n!}\,\dd^nF_\varepsilon/\dd\varepsilon^n|_{\varepsilon=0}$ the expansion coefficients of any quantity evaluated on it. Since $\Theta_0$ solves the field equations, only the second variation contributes at quadratic order,
\begin{equation}\label{S2-full}
S^{(2)}_{\rm global}
=\frac12\left.\frac{\delta^2 S_{\rm global}}{\delta\Theta^a\,\delta\Theta^b}\right|_{\Theta_0}
[\Theta^{(1),a},\Theta^{(1),b}].
\end{equation}
Using $N^{(0)}=1$ and $(N^i)^{(0)}=(N^i)^{(1)}=0$, we obtain}\footnote{
We assume that any boundary contribution from the environment at the fiducial cell takes the form $\int \dd t \oint_{\partial\mathcal V_0} \mathcal J^i(\Theta,\partial\Theta,\ldots)\,\dd \Sigma_i$ (as is the case for local tensor-field theories). Then, at quadratic order, for homogeneous perturbations this term reduces to $\int \dd t \, j^i(t) \oint_{\partial\mathcal V_0}\dd \Sigma_i=0$, since $\oint_{\partial\mathcal V_0}\dd \Sigma_i=0$.}
{\begin{multline}\label{S2-appendix}
 S_{\rm global}^{(2)}
= V_0 \int\dd t\,\biggl[
\pi^{(1),ij}\dot\gamma_{ij}^{(1)}+p_I^{(1)}\dot q^{(1),I}
\\
{}-\mathsf H-N^{(1)}\left.\frac{\partial\mathcal H}{\partial\Psi^a}\right|_{\Psi_0}
\Psi^{(1),a}\biggr],
\end{multline}
with the quadratic Hamiltonian
\begin{equation}\label{H-hessian}
\mathsf H\equiv\frac12\left.\frac{\partial^2\mathcal H}{\partial\Psi^a\,\partial\Psi^b}\right|_{\Psi_0}\Psi^{(1),a}\,\Psi^{(1),b}.
\end{equation}

Since $\dot\Theta_0=0$, the background Hamilton equations give
\begin{equation}\label{Henv-null-1}
\left.\frac{\partial\mathcal H}{\partial\Psi}\right|_{\Psi_0}=0,
\end{equation}
so the lapse term in Eq.~\eqref{S2-appendix} vanishes. Using $\pi^{(1),ij}\dot\gamma_{ij}^{(1)}=\pi_\varphi\dot\varphi$, the quadratic global action reduces to
\begin{equation}\label{S2-appendix-reduced}
S_{\rm global}^{(2)}=V_0\int\dd t\,\bigl(\pi_\varphi\dot\varphi+P_I\dot Q^I-\mathsf H\bigr),
\end{equation}
with $\Psi^{(1)}=(\varphi\,\delta_{ij},\tfrac13\pi_\varphi\delta^{ij},Q^I,P_I)$. Splitting $\mathsf H=\mathsf H_g+\mathsf H_{\rm env}$, with $\mathsf H_g=-\tfrac{\kappa}{3}\pi_\varphi^2$, yields Eq.~\eqref{eq:action-2}. In terms of $(\varphi,\pi_\varphi,Z)$, $\mathsf H_{\rm env}$ takes the form of Eq.~\eqref{eq:env-ham-2}.}

{
\prlendmattersection{Caldeira--Leggett form of the quadratic action}\label{app:cl-form}%
Adding the vacuum contribution to Eq.~\eqref{eq:env-ham-2}, the total quadratic Hamiltonian reads, in matrix notation,
\begin{multline}\label{H2-total}
    \mathsf H=\frac12\,\mathsf h_{\varphi\varphi}\,\varphi^2+\varphi\,\mathsf h_{\varphi\pi}\,\pi_\varphi+\frac{\pi_\varphi^2}{2M_{\rm eff}}\\
    +\varphi\,\mathsf h_{\varphi Z}\,Z+\pi_\varphi\,\mathsf h_{\pi Z}\,Z+\frac12\,Z^{\mathsf T}\mathsf h_{ZZ}\,Z,
\end{multline}
with $M_{\rm eff}$ defined in Eq.~\eqref{eq:Meff}. We denote by $\mathsf J$ the canonical symplectic matrix of the environmental variables, $\mathsf J^{AB}\equiv\{Z^A,Z^B\}$.

\prlendmatterheading{Canonical removal of momentum mixing}We define $\mathsf S\equiv M_{\rm eff}\,\mathsf J\mathsf h_{\pi Z}^{\mathsf T}$. The transformation
\begin{equation}\label{can-map}
    \Pi_\varphi=\pi_\varphi+M_{\rm eff}\bigl(\mathsf h_{\pi Z}Z+\mathsf h_{\varphi\pi}\,\varphi\bigr),\qquad Y=Z-\mathsf S\varphi,
\end{equation}
is canonical, since
\begin{equation}\label{bracket-check}
\{Y^A,\Pi_\varphi\}=\{Z^A,\Pi_\varphi\}-\mathsf S^A\{\varphi,\Pi_\varphi\}=\mathsf S^A-\mathsf S^A=0,
\end{equation}
and the remaining brackets are unchanged. It leaves the metric perturbation unchanged, so $\varphi$ keeps its geometric meaning. Substituting into Eq.~\eqref{H2-total} gives
\begin{equation}\label{H2-decoupled}
    \mathsf H=\frac12\,\widetilde{\mathsf h}_{\varphi\varphi}\,\varphi^2+\frac{\Pi_\varphi^2}{2M_{\rm eff}}+\varphi\,\widetilde{\mathsf h}_{\varphi Z}\,Y
    +\frac12\,Y^{\mathsf T}\widetilde{\mathsf h}_{ZZ}\,Y,
\end{equation}
where
\begin{equation}\label{tilde-coeffs}
\begin{aligned}
    \widetilde{\mathsf h}_{ZZ}&\equiv\mathsf h_{ZZ}-M_{\rm eff}\,\mathsf h_{\pi Z}^{\mathsf T}\mathsf h_{\pi Z},\\
    \widetilde{\mathsf h}_{\varphi Z}&\equiv\mathsf h_{\varphi Z}+\mathsf S^{\mathsf T}\mathsf h_{ZZ}-M_{\rm eff}\,\mathsf h_{\varphi\pi}\mathsf h_{\pi Z},\\
    \widetilde{\mathsf h}_{\varphi\varphi}&\equiv\mathsf h_{\varphi\varphi}+2\,\mathsf h_{\varphi Z}\mathsf S+\mathsf S^{\mathsf T}\mathsf h_{ZZ}\mathsf S-M_{\rm eff}\,\mathsf h_{\varphi\pi}^2.
\end{aligned}
\end{equation}

\prlendmatterheading{Canonical oscillator coordinates}Since $\Pi_\varphi=M_{\rm eff}\dot\varphi$, the fixed-geometry conditions $\varphi=\dot\varphi=0$ read $\varphi=\Pi_\varphi=0$, for which $\mathsf H=\frac12Y^{\mathsf T}\widetilde{\mathsf h}_{ZZ}Y$. The stability condition of Eq.~\eqref{eq:min} thus reads $\widetilde{\mathsf h}_{ZZ}>0$, so the environment admits canonical oscillator coordinates $(x_j,\pi_j)$ in which\footnote{Since $\widetilde{\mathsf h}_{ZZ}>0$, Williamson's theorem~\cite{Williamson1936} provides a linear canonical transformation bringing $\frac12Y^{\mathsf T}\widetilde{\mathsf h}_{ZZ}Y$ to a sum of oscillators, as in the first relation. This form is preserved by rotations in each phase plane $(\omega_jx_j,\pi_j)$, which can be chosen to remove the coupling of $\varphi$ to $\pi_j$.}
\begin{equation}\label{osc-form}
\begin{aligned}
    \frac12Y^{\mathsf T}\widetilde{\mathsf h}_{ZZ}Y&=\frac12\sum_j\bigl(\pi_j^2+\omega_j^2x_j^2\bigr),\qquad\omega_j>0,\\
    \varphi\,\widetilde{\mathsf h}_{\varphi Z}Y&=-\varphi\sum_j\alpha_jx_j.
\end{aligned}
\end{equation}

\prlendmatterheading{Caldeira--Leggett Hamiltonian}Substituting Eq.~\eqref{osc-form} into Eq.~\eqref{H2-decoupled} and completing the squares in $x_j$ yields $H_{\rm CL}$ in Eq.~\eqref{H_CL}, with $V(\varphi)=\sigma\varphi^2/2$ and
\begin{equation}\label{eq:sigma-em}
    \sigma=\widetilde{\mathsf h}_{\varphi\varphi}-\sum_j\frac{\alpha_j^2}{\omega_j^2}
    =\widetilde{\mathsf h}_{\varphi\varphi}-\widetilde{\mathsf h}_{\varphi Z}\,\widetilde{\mathsf h}_{ZZ}^{-1}\widetilde{\mathsf h}_{\varphi Z}^{\mathsf T}.
\end{equation}
}

\prlendmattersection{Diffeomorphism invariance and the effective potential}%
\label{app:covariance-stiffness}
{Spatial diffeomorphism invariance of the coupled action for the metric and the environmental DOFs retained implies \(\sigma=0\). Let $D_s$ be a one-parameter family of time-independent spatial diffeomorphisms, with $D_0=\mathrm{id}$. Transforming all background fields gives a family of static solutions $\Theta_s=D_s^*\Theta_0$. Differentiating the full field equations along this family shows that its tangent $\dd\Theta_s/\dd s|_{s=0}$ solves the linearized equations.

The dilation $D_s:x^i\mapsto e^{s/2}x^i$ preserves homogeneity and isotropy, and gives $\gamma_{ij}(s)=e^{s}\delta_{ij}$, so the tangent has $\varphi=1$. We write the oscillator coordinates of this static tangent\footnote{The canonical transformations used to obtain the Caldeira--Leggett variables are time independent, so the configuration remains static.} as $x_j=b_j$. Hamilton's equations from Eq.~\eqref{H_CL}, evaluated on this solution, give $\Pi_\varphi=\pi_j=0$ and
\begin{equation}\label{eq:cov-static}
 \begin{aligned}
 \dot\Pi_\varphi&=-\sigma+\sum_j\alpha_j\left(b_j-\frac{\alpha_j}{\omega_j^2}\right)=0,\\
 \dot\pi_j&=\alpha_j-\omega_j^2b_j=0.
 \end{aligned}
\end{equation}
The oscillator equations give $b_j=\alpha_j/\omega_j^2$, so the metric equation reduces to $\sigma=0$.

This cancellation can be relaxed by DOFs not retained in the environment, such as fermionic or nonlocal DOFs, when their influence is encoded in prescribed data $\bar B$. The action is then conditional, $S_{\bar B}[g,z]$, and diffeomorphism invariance takes the form
\begin{equation}\label{eq:cov-conditional}
 S_{f^*\bar B}[f^*g,f^*z]=S_{\bar B}[g,z].
\end{equation}
Therefore, only diffeomorphisms that leave $\bar B$ invariant remain symmetries of the conditional action, which generically invalidates the argument above.}

\prlendmattersection{Laplace analysis of the mean evolution}\label{app:laplace}%
We analyze here the dynamics determined by Eq.~\eqref{eq:varphi-langevin} using Laplace methods{, assuming $M_{\rm eff}<0$ throughout this and the following End Matter sections}. For a suitable real function $f(\eta)$, we denote its Laplace transform by
\begin{equation}
    \widetilde f(s) = \int_0^\infty f(\eta)\, e^{-s \eta}\, \dd \eta.
\end{equation}

In Laplace space, the evolution equation for the metric perturbation, Eq.~\eqref{eq:varphi-langevin}, reads
\begin{equation}\label{eq:varphi(z)}
    \widetilde{\varphi} (s) =\frac{-{M_{\rm eff}}(s\varphi_0+ \dot \varphi_0) -\widetilde{\gamma}(s)\varphi_0 - \widetilde{\xi} (s)}{-{M_{\rm eff}}s^2-s\gtilde(s)-\sigma},
\end{equation}
where $\widetilde{\xi}$ is the pathwise Laplace transform of the Gaussian noise process $\xi$. We restrict the analysis in this End Matter section to the mean evolution, for which $\langle\widetilde{\xi}\rangle = 0$. The mean perturbation $\langle\widetilde{\varphi}(s)\rangle$ is then fixed by the initial conditions, $\sigma$, and the damping kernel (or through Eq.~\eqref{eq:gamma-J-rel}, by the spectral density). From here on, abusing notation, we write $\widetilde{\varphi}(s)$ for the transform of the mean perturbation.

The damping-kernel transform follows from Eq.~\eqref{eq:gamma-J-rel},
\begin{equation}\label{eq:dem-tilde-gamma}
    \widetilde\gamma(s) = \frac{2}{\pi} \int_0^\infty \frac{J(\omega)}{\omega}\frac{s}{s^2+\omega^2}\,\dd\omega, \qquad \Re s>0,
\end{equation}
holomorphic in the right half-plane.

\medskip
Since $\gtilde$ is holomorphic in $\Re s > 0$ (see Eq.~\eqref{eq:dem-tilde-gamma}), Eq.~\eqref{eq:varphi(z)} shows that $\widetilde{\varphi}(s)$ is meromorphic there. Any pole $p \in \mathbb{C}$ with $\Re p > 0$ then contributes an exponentially divergent term $\Re(e^{p\eta})$ to the evolution of $\varphi(\eta)$. These poles are the zeros of the denominator of Eq.~\eqref{eq:varphi(z)},
\begin{equation}\label{eq:D(s)}
    D(s) \equiv -{M_{\rm eff}}s^2-s\gtilde(s)-\sigma,
\end{equation}
that are not cancelled by a coincident zero of the numerator. Such cancellation can occur only for spectral densities whose $D$ has a real simple zero $\lambda > 0$, and provided the initial conditions satisfy
\begin{equation}\label{eq:in-cond-cancellation}
    \dot\varphi_0 = \frac{\sigma}{{M_{\rm eff}}\lambda}\,\varphi_0.
\end{equation}

We split the analysis of the poles of $\widetilde{\varphi}$ into three cases according to the sign of $\sigma$.

\prlendmatterheading{$\sigma > 0$}%
In this case, note that
\begin{align}
   \lim_{x \downarrow 0,\, x\in\mathbb{R}} D(x) &= -\sigma < 0, \\
   \lim_{x \to \infty,\, x\in\mathbb{R}} D(x) &= +\infty. \label{eq:asymp-D}
\end{align}
Since $D(x)$ is continuous, there exists $\lambda > 0$ with $D(\lambda) = 0$. Unless the initial data satisfy the cancellation condition of Eq.~\eqref{eq:in-cond-cancellation}, this zero is a pole of $\widetilde\varphi$ and produces an exponentially divergent contribution. Thus the evolution is generically unbounded.

\prlendmatterheading{$\sigma = 0$}%
With $\sigma = 0$, Eq.~\eqref{eq:varphi(z)} reduces to
\begin{equation}\label{eq:varphi-sigma0}
    \widetilde\varphi(s) = \frac{\varphi_0}{s} + \frac{{M_{\rm eff}}\dot\varphi_0}{s^2\left({M_{\rm eff}} + \gtilde(s)/s\right)}.
\end{equation}
If $\dot\varphi_0=0$, this gives the constant solution $\varphi(\eta)=\varphi_0$, corresponding to Minkowski up to a rescaling of the scale factor. For $\dot\varphi_0\neq0$, the evolution is unbounded, as we show now.

For $x > 0$, define
\begin{equation}
    K(x) \equiv \frac{\gtilde(x)}{x} = \frac{2}{\pi} \int_0^\infty \frac{J(\omega)/\omega}{x^2 + \omega^2}\, \dd\omega,
\end{equation}
the second equality following from Eq.~\eqref{eq:gamma-J-rel} (see also Eq.~\eqref{eq:dem-tilde-gamma}). Since $J(\omega) \geq 0$, we have $K'(x) \leq 0$, so $K$ is non-increasing and the limit
\begin{equation}
    K_0 \equiv \lim_{x\downarrow 0} K(x) \in [0,\infty]
\end{equation}
is well defined. We distinguish two cases.

If $K_0 \leq -{M_{\rm eff}}$, then
\begin{equation}
    \lim_{x \downarrow 0} x\, \widetilde\varphi(x) = \varphi_0 + \lim_{x\downarrow 0} \frac{{M_{\rm eff}}\dot\varphi_0}{x\left({M_{\rm eff}} + K(x)\right)},
\end{equation}
whose second term diverges. Hence $\widetilde\varphi(x) \neq \mathcal{O}(1/x)$ as $x \downarrow 0$, and $\varphi$ is unbounded.

If instead $K_0 > -{M_{\rm eff}}$, then since $K(x) \to 0$ as $x \to \infty$ and $K$ is continuous, there exists $\lambda > 0$ with $K(\lambda) = -{M_{\rm eff}}$. This is a real pole of $\widetilde\varphi$ in the right half-plane, again producing an exponentially divergent contribution.

\prlendmatterheading{$\sigma < 0$}%
A positive real zero of $D$ exists whenever
\begin{equation}\label{eq:sigma-neg-cond}
    \inf_{x > 0, \,  x\in \mathbb{R}} \left[-{M_{\rm eff}}x^2 - x\, \gtilde(x)\right] < \sigma,
\end{equation}
which, by Eq.~\eqref{eq:gamma-J-rel}, is fulfilled by a spectral density of sufficiently large amplitude, that is, in the strong-coupling regime.

In the weak-coupling regime a pole in the right half-plane is also present for any spectral density bounded between two positive constants around the natural frequency of the harmonic potential, ${\Omega_0=\sqrt{\sigma/{M_{\rm eff}}}}$ (see Theorem \ref{theorem}).

\prlendmattersection{Conditions for bounded evolution}\label{app:conditions}%
Proofs of the results stated in this and the following End Matter section are collected in the Supplemental Material~\cite{SM}.

\begin{theorem}\label{theorem}
Let $\sigma<0$, $\zeta>0$ and consider a spectral density of the form $J_\zeta(\omega) = \zeta\, J(\omega)$, with $J/\omega\in L^1(0,\infty)$. Assume that $J$ is essentially bounded between two positive constants near ${\Omega_0=\sqrt{\sigma/{M_{\rm eff}}}}$. Then there exists $\zeta_{\mathrm u}>0$ such that, for every $0<\zeta<\zeta_{\mathrm u}$, the associated metric perturbation in Laplace space, $\widetilde{\varphi_\zeta}$, has a pole $s_\zeta$ with $\Re \, (s_\zeta)>0$ for all nonzero initial data $(\varphi_0,\dot \varphi_0)$.
\end{theorem}

\begin{theorem}\label{thm:bounded}
    Let $\sigma<0$, let $0<\omega_c<{\Omega_0\equiv\sqrt{\sigma/{M_{\rm eff}}}}$, and consider the spectral density
    \begin{equation}\label{eq:J-general}
        J_\zeta(\omega)\equiv\zeta\,j(\omega)\,\mathbf 1_{(0,\omega_c)}(\omega),
        \qquad \zeta>0,
    \end{equation}
    with $j\in C^{0,\beta}([0,\omega_c])$ for some $\beta>0$, $j(0)=0$, and $j>0$ on $(0,\omega_c]$. Set
    \begin{equation}
        \mathcal G\equiv\int_0^{\omega_c}\frac{j(\omega)}{\omega}\,\dd\omega,
        \qquad
        \zeta_{\mathrm b}\equiv-\frac{\pi{M_{\rm eff}}(\Omega_0-\omega_c)^2}{2\mathcal G}.
    \end{equation}
    Then every $\zeta\in(0,\zeta_{\mathrm b})$ yields a bounded metric perturbation of the form of Eq.~\eqref{eq:final-thm}, with real constants $A_k,B_k$, frequencies $\omega_c<\nu_1<\nu_2<\Omega_0$, and $\varphi_{\zeta,\mathrm{tr}}(\eta)\to0$ as $\eta\to\infty$.
\end{theorem}

\prlendmattersection{Power-law spectral densities with a sharp cutoff}\label{app:powerlaw}%
The spectral density of Eq.~\eqref{eq:J-sharp} satisfies the hypotheses of Theorem~\ref{thm:bounded} with $j(\omega)=\omega^{p}$. The decay laws of the resulting transient $\varphi_{\zeta,\mathrm{tr}}(\eta)$ in Eq.~\eqref{eq:final-thm}, displayed in Fig.~\ref{fig:bounded-ev}(b), are governed by the endpoints $\omega=0$ and $\omega=\omega_c$ of its spectral representation, the interior contributing faster-decaying terms; they are stated in the following theorem and derived in the Supplemental Material~\cite{SM}.

\medskip
Define, for $0<\omega<\omega_c$,
\begin{multline}\label{eq:A-powerlaw}
A_\zeta(-\omega^2)
\equiv -{M_{\rm eff}}(\Omega_0^2-\omega^2)\\
+\frac{2\zeta}{\pi}\,\omega^2\,\operatorname{p.v.}\!\int_0^{\omega_c}
\frac{(\omega')^{p-1}}{(\omega')^2-\omega^2}\,\dd\omega'.
\end{multline}

\begin{theorem}\label{thm:X}
Let $\sigma<0$, $0<\omega_c<\Omega_0$, $0<\zeta<\zeta_{\mathrm b}$, and $(\varphi_0,\dot\varphi_0)\neq(0,0)$. Define
\begin{equation}
    R_0 \equiv \Big[ \dot\varphi_0^{2} + \left( \tfrac{\Omega_0^{2}\varphi_0}{\omega_c} \right)^{2} \Big]^{1/2}, \; \delta \equiv \operatorname{atan2} \left( \tfrac{\Omega_0^{2}\varphi_0}{\omega_c}, \dot\varphi_0 \right),
\end{equation}
and
\begin{equation}
a_c(\eta)\equiv-{\frac{2\pi{M_{\rm eff}}R_0}{\zeta\,\omega_c^{\,p}\,\eta\,\bigl[\ln^{2}(\eta/\eta_*)+\pi^{2}\bigr]}},
\label{X:ac}
\end{equation}
with $\eta_*\equiv\omega_c^{-1}e^{\pi\beta_c/(\zeta\omega_c^{p})}$,
\begin{equation}\label{X:betac}
\beta_c\equiv\lim_{\omega\uparrow\omega_c}\Bigl[A_\zeta(-\omega^{2})-\frac{\zeta}{\pi}\,\omega^{p}\ln\Bigl(\frac{\omega_c-\omega}{\omega_c}\Bigr)\Bigr].
\end{equation}

Then, as $\eta\to\infty$,
\begin{multline}\label{X:main}
\varphi_{\zeta,\mathrm{tr}}(\eta)=s(\eta)
+a_c(\eta)\cos(\omega_c\eta+\delta)
\\
+\Oh \big(\eta^{-1}\ln^{-3}(\eta)\big),
\end{multline}
where, with $q=\min\{p,2\}$,
\begin{equation}\label{X:s}
    \begin{aligned}
    s(\eta)
    ={}&
    {\frac{2\zeta\Gamma(p)\cos(\pi p/2)}
    {\pi{M_{\rm eff}}\Omega_0^2}}\,
    \varphi_0\,\eta^{-p}\\
    &+
    {\frac{2\zeta\Gamma(p+1)\cos(\pi p/2)}
    {\pi{M_{\rm eff}}\Omega_0^4}}\,
    \dot\varphi_0\,\eta^{-p-1}\\
    &+
    \varphi_0\,\Oh\!\left(\eta^{-p-q}\ln\eta\right)\\
    &+
    \dot\varphi_0\,\Oh\!\left(\eta^{-p-q-1}\ln\eta\right).
    \end{aligned}
\end{equation}
\end{theorem}

These are the dashed decay rates shown in Fig.~\ref{fig:bounded-ev}(b): 
for $\varphi_0 \neq 0$ and $0 < p < 1$ (sub-Ohmic) the transient decays as  $\eta^{-p}$, whereas for $p \geq 1$, or for $\varphi_0 = 0$, the oscillatory contribution $a_c(\eta)\cos(\omega_c\eta+\delta)$ dominates, so the r.m.s.\ decays as $\Oh ([\eta\ln^2(\eta/\eta_*)]^{-1})$. Indeed, for $p=1$ one has $q=1$ and $\cos(\pi/2)=0$, so, by Eq.~\eqref{X:s}, $s(\eta)=\Oh(\eta^{-2}\ln\eta) = o\big(\eta^{-1}\ln^{-3}(\eta)\big)$, and this contribution can be absorbed into the remainder; for $p>1$, $\varphi_0\neq0$, Eq.~\eqref{X:s} gives $s(\eta)=\Oh(\eta^{-p}) =o\big(\eta^{-1}\ln^{-3}(\eta)\big)$. If $\varphi_0=0$, Eq.~\eqref{X:s} instead gives $s(\eta)=\Oh(\eta^{-p-1})=o\big(\eta^{-1}\ln^{-3}(\eta)\big)$ for every $p>0$, so the oscillatory contribution always dominates.

\nocite{Muskhelishvili1953,Bingham:1989rv}

\onecolumngrid
\vspace{5\baselineskip}
\begin{center}
    \textbf{\large Supplemental Material}
\end{center}
\vspace{0.15\baselineskip}
\twocolumngrid
\setcounter{equation}{0}
\setcounter{theorem}{0}
\renewcommand{\theequation}{S\arabic{equation}}
\renewcommand{\thetheorem}{S\arabic{theorem}}
\renewcommand{\theHequation}{S.\arabic{equation}}
\renewcommand{\theHtheorem}{S.\arabic{theorem}}
\input{supplemental-material}

\end{document}

%% file: supplemental-material.tex
This Supplemental Material collects the proofs of the results stated in the End Matter of the Letter. {Its two sections follow End Matter Secs.~\ref{app:conditions} and \ref{app:powerlaw}. The first proves Theorems~\ref{theorem} and \ref{thm:bounded}, and the second proves Theorem~\ref{thm:X}, using an auxiliary lemma.} Throughout, equation, theorem, and figure numbers without an ``S'' prefix refer to the Letter; references with an ``S'' prefix refer to this Supplemental Material.

\section{Conditions for bounded evolution}

\begin{proof}[Proof of Theorem~\ref{theorem}]
    The assumption on $J$ means that there exist $J_+>J_->0$ and $0<r<\Omega_0$ such that $J_-<J(\omega)<J_+$ for a.e. $\omega \in (\Omega_0-r, \Omega_0+r)$.

    \smallskip
    Now, using Eq.~\eqref{eq:gamma-J-rel}, $\widetilde{\gamma_\zeta}$ can be expressed as
    \begin{equation}
        \widetilde{\gamma_\zeta}(s) =  \frac{2}{\pi} \zeta \int_0^\infty \frac{J(\omega)}{\omega}\frac{s}{s^2+\omega^2}\,\dd\omega,
        \qquad \Re s>0.
    \end{equation}

    Therefore, $D_\zeta$ in Eq.~\eqref{eq:D(s)} reads
    \begin{equation}\label{eq:dem-D}
        D_\zeta(s)
        =
        -{M_{\rm eff}}(s^2+\Omega_0^2)
        -\frac{2}{\pi} \zeta\int_0^\infty \frac{J(\omega)}{\omega}\frac{s^2}{s^2+\omega^2}\,\dd\omega.
    \end{equation}

    As $D_\zeta$ depends on $s^2$, we define
    \begin{equation}\label{eq:dem-F}
        F_\zeta(z)
        \equiv
        -{M_{\rm eff}}(z+\Omega_0^2)
        -\frac{2}{\pi} \zeta\int_0^\infty  \frac{J(\omega)}{\omega}\frac{z}{z+\omega^2}\,\dd\omega.
    \end{equation}
    The poles sought correspond to zeros of $D_\zeta(s) = F_\zeta(s^2)$ in the right half-plane. Without loss of generality, we look for zeros of $F_\zeta(z)$ expressed as
    \begin{equation}\label{eq:dem-z}
        z=-u+iv,
        \qquad u>0,\quad v>0.
    \end{equation}
    To find such poles, we analyze $\Im F_\zeta(z) = 0$ and $\Re F_\zeta(z) = 0$ separately below.

    \medskip
    $\blacktriangleright$ $\Im F_\zeta(z) = 0$:

    \smallskip
    Using Eqs.~\eqref{eq:dem-F} and \eqref{eq:dem-z}, this equation is expressed as
    \begin{equation}\label{eq:dem-imeq}
        \mathcal I (u,v)=-\frac{\pi{M_{\rm eff}}}{2\zeta},
    \end{equation}
    where
    \begin{equation}
        \mathcal I(u,v) \equiv \int_0^\infty \frac{J(\omega)}{\omega}\frac{\omega^2}{(\omega^2-u)^2+v^2}\,\dd\omega.
    \end{equation}

    Note that $v \mapsto \mathcal{I}(u,v)$ is continuous and strictly decreasing in $(0,\infty)$. Moreover, for every fixed $u$,
    \begin{equation}
        \lim_{v \to \infty} \mathcal I (u,v) = 0.
    \end{equation}

    As shown below, for fixed $u\in I$ with
    \begin{equation}
        I=[u_-,u_+],
        \qquad
        u_\pm
        \equiv
        \left(\Omega_0\pm\frac r2\right)^2,
    \end{equation}
    we also have
    \begin{equation}
        \lim_{v \downarrow 0} \mathcal I (u,v) = +\infty.
    \end{equation}

    Indeed, define $t\equiv\omega^2-u$ and take $\tau>0$ such that $\sqrt{u+t}\in(\Omega_0-r,\Omega_0+r)$ for every $u\in I$, $|t| \leq \tau$. Then, 
    \begin{equation}\label{eq:dem-I-lower-bound}
        \begin{aligned}
            \mathcal I(u,v)
            &\ge
            \frac{J_-}2
            \int_{-\tau}^{\tau}
            \frac{\dd t}{t^2+v^2} \\
            &=
            \frac{J_-}{v}
            \arctan\frac{\tau}{v},
            \end{aligned}
    \end{equation}
    which tends to $+\infty$ as $v\downarrow 0$.

    Therefore, we conclude that, for every $u\in I$, Eq.~\eqref{eq:dem-imeq} has a unique solution $v_\zeta(u)> 0$.

    Before solving $\Re F_\zeta(z) = 0$ we need to establish the following bounds for $v_\zeta(u)$.

    \medskip
    $\circ $ $c\, \zeta\le v_\zeta(u)\le C\,\zeta$, for all $u\in I$ and certain $c, C>0$:

    \smallskip
    Defining $t=\omega^2-u$,
    \begin{equation}
        \begin{aligned}
            \int_{\Omega_0-r}^{\Omega_0+r}
            \frac{J(\omega)\omega}
                {(\omega^2-u)^2+v^2}\,\dd\omega
            &\le
            \frac{J_+}2
            \int_{\mathbb R}\frac{\dd t}{t^2+v^2} \\
            &=
            \frac{\pi J_+}{2v},
        \end{aligned}
    \end{equation}
    and also, since \(|\omega^2-u|\) is uniformly bounded away from zero for \(u\in I\),
    \begin{equation}
        \int_{(0,\infty) \setminus (\Omega_0-r, \Omega_0+r)} \frac{J(\omega)\omega}
                {(\omega^2-u)^2+v^2}\,\dd\omega \leq C_2.
    \end{equation}

    Consequently, for $u \in I$ and certain constants $C_1, C_2>0$,
    \begin{equation}\label{eq:dem-I-ineq}
        \mathcal I(u,v)\le\frac{C_1}{v}+C_2.
    \end{equation}

    Applying Eq.~\eqref{eq:dem-I-ineq} at $v=v_\zeta(u)$ and using Eq.~\eqref{eq:dem-imeq},
    \begin{equation}
        -\frac{\pi{M_{\rm eff}}}{2\zeta}
        \le
        \frac{C_1}{v_\zeta(u)}+C_2,
    \end{equation}
    which, for sufficiently small $\zeta$, such that $-\frac{\pi{M_{\rm eff}}}{2\zeta}-C_2 \ge -\frac{\pi{M_{\rm eff}}}{4\zeta}$, translates into
    \begin{equation}\label{eq:dem-v-up}
        v_\zeta(u)\le C\zeta
    \end{equation}
    uniformly for \(u\in I\).

    Now, for sufficiently small $\zeta$, Eq.~\eqref{eq:dem-v-up} ensures that \(v_\zeta(u)\le\tau\), and by Eq.~\eqref{eq:dem-I-lower-bound},
    \begin{equation}
        \mathcal I(u,v_\zeta(u))
        \ge
        \frac{\pi}{4}\frac{ J_-}
            {v_\zeta(u)}.
    \end{equation}

    Thus, using Eq.~\eqref{eq:dem-imeq} we conclude that
    \begin{equation}\label{eq:dem-v-zeta-lower-bound}
        v_\zeta(u)\ge c\zeta
    \end{equation}
    uniformly for \(u\in I\).

    Note that, in particular, $c\, \zeta\le v_\zeta(u)\le C\,\zeta$ and continuity of $\mathcal{I}$ imply that \(u\mapsto v_\zeta(u)\) is continuous.

    \medskip
    $\blacktriangleright$ $\Re F_\zeta(z) = 0$:

    \smallskip
    Using Eqs.~\eqref{eq:dem-F} and \eqref{eq:dem-z}, this equation is expressed as
    \begin{equation}\label{eq:dem-reeq}
        (\Omega_0^2-u) + \frac{2\zeta}{\pi{M_{\rm eff}}}\mathcal A(u,v) = 0,
    \end{equation}
    where
    \begin{equation}
        \mathcal A(u,v) \equiv \Re \int_0^\infty \frac{J(\omega)}{\omega} \frac{-u+iv}{\omega^2-u+iv}\,\dd\omega.
    \end{equation}

    We next show that
    \begin{equation}\label{eq:dem-A-limit}
        \lim_{\zeta\downarrow0} \, \sup_{u\in I} \zeta|\mathcal A(u,v_\zeta(u))| = 0.
    \end{equation}

    Using $\frac{z}{z+\omega^2} = 1-\frac{\omega^2}{z+\omega^2},$ we obtain
    \begin{multline}\label{eq:dem-A-upperbound}
        |\mathcal A(u,v)| \le \left\|\frac{J(\omega)}{\omega}\right\|_{L^1}\\+ \int_0^\infty \frac{J(\omega)\omega}{\sqrt{(\omega^2-u)^2+v^2}}\,\dd\omega.
    \end{multline}

    For the integral term in Eq.~\eqref{eq:dem-A-upperbound}, note that the integral outside \((\Omega_0-r,\Omega_0+r)\) is uniformly bounded for \(u\in I\). Inside this interval, the substitution \(t=\omega^2-u\) gives
    \begin{multline}
        \int_{\Omega_0-r}^{\Omega_0+r}
        \frac{J(\omega)\omega}
            {\sqrt{(\omega^2-u)^2+v^2}}\,\dd\omega\\
        \le
        \frac{J_+}2
        \int_{-T}^{T}
        \frac{\dd t}{\sqrt{t^2+v^2}} = J_+ \operatorname{arsinh}\frac{T}{v},
    \end{multline}
    for sufficiently large $T>0$.

    Therefore, using that $\operatorname{arsinh}\frac{T}{v} \leq C\left(1+\ln\frac1v\right)$ for $0<v\le1$, we obtain, for $u\in I$ and $0<v\le1$,
    \begin{equation}
        |\mathcal A(u,v)|\le C\left(1+\ln\frac1v\right),
    \end{equation}
    from which Eq.~\eqref{eq:dem-A-limit} follows.

    Now, to solve Eq.~\eqref{eq:dem-reeq}, we define for $u\in I$,
    \begin{equation}
        R_\zeta(u) \equiv (\Omega_0^2-u) + \frac{2\zeta}{\pi{M_{\rm eff}}}\mathcal A(u,v_\zeta(u)).
    \end{equation}

    Since $\Omega_0^2-u_->0$ and $\Omega_0^2-u_+<0$, for sufficiently small $\zeta$
    \begin{equation}
        R_\zeta(u_-)>0, \qquad R_\zeta(u_+)<0,
    \end{equation}
    which, given that $R_\zeta$ is continuous, gives \(u_\zeta\in I\) such
    that
    \begin{equation}
        R_\zeta(u_\zeta)=0.
    \end{equation}

    \medskip
    Defining $z_\zeta \equiv -u_\zeta + i\,v_\zeta(u_\zeta)$, we have $F_\zeta(z_\zeta)=0$, so the square roots of $z_\zeta$ are zeros of $D_\zeta$. Since $u_\zeta, v_\zeta(u_\zeta)>0$, $z_\zeta$ lies in the upper-left quadrant, and we let $s_\zeta$ be its square root in the first quadrant. Then $\Re s_\zeta, \Im s_\zeta > 0$, so $s_\zeta$ is a zero of $D_\zeta$ in the right half-plane that is not purely real. It therefore cannot be cancelled by a zero of the numerator in Eq.~\eqref{eq:varphi(z)} provided $(\varphi_0, \dot \varphi_0) \neq (0, 0)$, and hence $s_\zeta$ is a pole of $\widetilde{\varphi_\zeta}$.
\end{proof}

\begin{proof}[Proof of Theorem~\ref{thm:bounded}]
    We use $D_\zeta(s)$ and $F_\zeta(z)$ from Eqs.~\eqref{eq:dem-D} and \eqref{eq:dem-F}, related by $D_\zeta(s)=F_\zeta(s^2)$. For the spectral densities of Eq.~\eqref{eq:J-general}, the split $z/(z+\omega^2)=1-\omega^2/(z+\omega^2)$ isolates the singular part of $F_\zeta$,
    \begin{equation}\label{eq:F-def}
        F_\zeta(z)= -{M_{\rm eff}}(z+\Omega_0^2)-\frac{2}{\pi}\zeta \big[ {\mathcal G} - \Phi(z)\big],
    \end{equation}
    with
    \begin{equation}\label{eq:Phi}
        \Phi(z)\equiv \int_0^{\omega_c}\frac{\omega\,j(\omega)}{z+\omega^2}\,\dd \omega = \frac12\int_0^{\omega_c^2} \frac{j(\sqrt{u})}{u-u_z}\,\dd u,
     \end{equation}
    with $u\equiv\omega^2$ and $u_z\equiv-z$.

    The kernel $\Phi$ is singular only for $z\in(-\omega_c^2,0)$. Hence $F_\zeta$ is holomorphic on $\mathbb C\setminus[-\omega_c^2,0]$, and $D_\zeta$ is holomorphic on $\mathbb C\setminus i[-\omega_c,\omega_c]$.

    The metric perturbation for $\eta>0$ follows from the inversion
    \begin{equation}\label{eq:inv-laplace}
        \varphi_\zeta(\eta)=\frac{1}{2\pi i}\int_{\varsigma-i\infty}^{\varsigma+i\infty}
        e^{s\eta}\,\widetilde\varphi_\zeta(s)\,\dd s,
        \quad \varsigma>0.
    \end{equation}

    We evaluate it by closing the Bromwich contour to the left and applying the residue theorem, in three steps: (i) determine the boundary values of $D_\zeta$ across the cut $i[-\omega_c,\omega_c]$; (ii) locate the zeros of $D_\zeta$ off the cut, which supply the enclosed poles; (iii) read off $\varphi_\zeta$ as a residue sum plus a cut integral.

    \medskip
    \emph{Boundary values across the cut.}
    The two faces of the cut of $D_\zeta$ are $\Re s\gtrless0$, so for $\omega\in(-\omega_c,\omega_c)$ we set
    \begin{equation}
        D_{\zeta,\pm}(i\omega)\equiv\lim_{\epsilon\downarrow0}D_\zeta(\pm\epsilon+i\omega).
    \end{equation}

    For the faces of $F_\zeta$ across $[-\omega_c^2,0]$, whose two sides are $\Im z\gtrless0$, we set
    \begin{equation}
        F_{\zeta,\pm}(x)\equiv\lim_{\epsilon\downarrow0}F_\zeta(x\pm i\epsilon),
        \qquad x\in(-\omega_c^2,0).
    \end{equation}

    Using that $D_\zeta(s)=F_\zeta(s^2)$, for $0<\omega<\omega_c$ these boundary values are related\footnote{By the Plemelj formula \cite{Muskhelishvili1953}, the boundary value $F_{\zeta,\pm}(x)$ is independent of how $z$ approaches $x$ within the upper (lower) half-plane; see below.} by $D_{\zeta,\pm}(i\omega)=F_{\zeta,\pm}(-\omega^2)$.

    Because $j$ is H\"older continuous, the Plemelj formula \cite{Muskhelishvili1953} applied to Eq.~\eqref{eq:Phi} gives, for every $x\in(-\omega_c^2,0)$,
    \begin{equation}
        \Phi_\pm(x)=\frac12\operatorname{p.v.}\!\int_0^{\omega_c^2}\frac{j(\sqrt u)}{u-u_x}\,\dd u
        \mp\frac{i\pi}{2}\,j(\sqrt{u_x}),
    \end{equation}
    so, by Eq.~\eqref{eq:F-def},
    \begin{equation}\label{eq:F-lim}
        F_{\zeta,\pm}(x)=A_\zeta(x)\mp i\zeta\,j(\omega_x),
        \qquad \omega_x\equiv\sqrt{-x},
    \end{equation}
    with
    \begin{multline}\label{eq:A}
        A_\zeta(x)=-{M_{\rm eff}}(x+\Omega_0^2)\\
        -\frac{2}{\pi}\zeta\left({\mathcal G}-\frac12\operatorname{p.v.}\!\int_0^{\omega_c^2}\frac{j(\sqrt u)}{u-u_x}\,\dd u\right).
    \end{multline}
    Thus $\Im F_{\zeta,\pm}(x)=\mp\zeta j(\omega_x)\neq0$ on $(-\omega_c^2,0)$, and neither face vanishes there. At the endpoints,
    \begin{equation}\label{eq:A-lim}
        \lim_{x \, \uparrow \, 0} A_\zeta(x) = -{M_{\rm eff}}\Omega_0^2, \qquad \lim_{x \, \downarrow \, -\omega_c^2} A_\zeta(x) = -\infty,
    \end{equation}
    since the principal value tends to ${\mathcal G}$ as $x\uparrow0$  and to $-\infty$ as $x\downarrow-\omega_c^2$.

    \medskip
    \emph{Zeros of $F_\zeta$ off the cut.}
    The poles of $\widetilde\varphi_\zeta$ off $i[-\omega_c,\omega_c]$ are the zeros of $D_\zeta$, equivalently the zeros of $F_\zeta$ off $[-\omega_c^2,0]$. We first exhibit two real zeros and then show they are the only ones. For $x<-\omega_c^2$, differentiating Eq.~\eqref{eq:F-def} twice under the integral,
    \begin{equation}\label{eq:concave}
        F_\zeta''(x)=\frac{4\zeta}{\pi}\int_0^{\omega_c}
        \frac{\omega\,j(\omega)}{(x+\omega^2)^3}\,\dd\omega<0 ,
    \end{equation}
    so $F_\zeta$ is strictly concave on $(-\infty,-\omega_c^2)$, with $F_\zeta(x)\to-\infty$ both as $x\to-\infty$ and as $x\uparrow-\omega_c^2$. At $x_*\equiv-\Omega_0\omega_c\in(-\Omega_0^2,-\omega_c^2)$ the bound $x_*/(x_*+\omega^2)\le\Omega_0/(\Omega_0-\omega_c)$ for $\omega\in(0,\omega_c)$ gives
    \begin{equation}\label{eq:xstar}
        F_\zeta(x_*)\ge
        \frac{\Omega_0}{\Omega_0-\omega_c}
        \left(-{M_{\rm eff}}(\Omega_0-\omega_c)^2-\frac{2\zeta}{\pi}{\mathcal G}\right)>0,
    \end{equation}
    the last inequality being equivalent to $\zeta<\zeta_{\mathrm b}$, since
    \begin{equation}
        {
        \zeta_{\mathrm b}
        =-\frac{\pi{M_{\rm eff}}(\Omega_0-\omega_c)^2}{2\mathcal G}.}
    \end{equation}
    A concave function negative at both ends of $(-\infty,-\omega_c^2)$ and positive at $x_*$ has exactly two simple zeros there. Since
    \begin{equation}
        F_\zeta(-\Omega_0^2)
        =-\frac{2\zeta}{\pi}\int_0^{\omega_c}\frac{j(\omega)}{\omega}\,
        \frac{\Omega_0^2}{\Omega_0^2-\omega^2}\,\dd\omega<0 ,
    \end{equation}
    they satisfy
    \begin{equation}\label{eq:F-roots}
        \xi_2\in(-\Omega_0^2,x_*),\qquad \xi_1\in(x_*,-\omega_c^2).
    \end{equation}

    These are the only zeros in $\mathbb C\setminus[-\omega_c^2,0]$. Apply the argument principle on the cycle made of $|z|=R$ (counterclockwise) and a $\delta$-neighborhood of $[-\omega_c^2,0]$ (clockwise), on which $F_\zeta\neq0$ for $R$ large and $\delta$ small by Eqs.~\eqref{eq:F-def} and \eqref{eq:F-lim}--\eqref{eq:A-lim}. The variation of $\arg F_\zeta$ is
    \begin{itemize}
        \item outer circle: ${F_\zeta(z)=-{M_{\rm eff}}z+\Oh(1)}$, so $\Delta\arg=2\pi$,
        \item upper face, $x$ from $-\omega_c^2$ to $0$: by Eq.~\eqref{eq:F-lim} the image lies in the lower half-plane and runs from $\arg\to-\pi$ (where $A_\zeta\to-\infty$) to $\arg=0$ (where $F_\zeta\to-{M_{\rm eff}}\Omega_0^2$), so $\Delta\arg=\pi$,
        \item lower face: by $F_\zeta(\bar z)=\overline{F_\zeta(z)}$ this is the mirror image, $\Delta\arg=\pi$;
        \item the end arcs at $0$ and $-\omega_c^2$ contribute $o(1)$, since $F_\zeta\to-{M_{\rm eff}}\Omega_0^2\neq0$ near $0$ and $|F_\zeta|\to\infty$ with bounded $\Im F_\zeta$ near $-\omega_c^2$.
    \end{itemize}
    Therefore $\tfrac1{2\pi}\Delta\arg F_\zeta=2$, so the zeros in Eq.~\eqref{eq:F-roots} are the only ones. Setting $\nu_k\equiv\sqrt{-\xi_k}\in(\omega_c,\Omega_0)$, the zeros of $D_\zeta(s)=F_\zeta(s^2)$ off $i[-\omega_c,\omega_c]$ are then $s=\pm i\nu_1,\pm i\nu_2$. They are simple because $D_\zeta'(s)=2sF_\zeta'(s^2)$ with $s\neq0$ and $F_\zeta'(\xi_k)\neq0$.

    \medskip
    \emph{Inversion.}
    Let $\Gamma_{R,\delta}$ be the positively oriented boundary of
    \begin{equation}\label{eq:region}
        \begin{split}
            \mathcal R_{R,\delta}\equiv{}
            &\bigl\{\,|s|<R,\ \Re s<\varsigma\,\bigr\}\\
            &\setminus\bigl\{\,\operatorname{dist}\!\bigl(s,\,i[-\omega_c,\omega_c]\bigr)\le\delta\,\bigr\}.
        \end{split}
    \end{equation}
    It comprises the Bromwich segment $\Re s=\varsigma$ (upward), the arc $|s|=R$ closing it counterclockwise through $\Re s<0$, and a keyhole $\mathcal C_\delta$ around the cut $i[-\omega_c,\omega_c]$ (clockwise). Take $R$ large and $\delta$ small, so that $\Gamma_{R,\delta}$ encloses the poles $\pm i\nu_1,\pm i\nu_2$.

    By the residue theorem,
    \begin{multline}\label{eq:res-finite}
        \frac{1}{2\pi i}\oint_{\Gamma_{R,\delta}}e^{s\eta}\widetilde\varphi_\zeta(s)\,\dd s\\
        =\sum_{p\in\{\pm i\nu_1,\pm i\nu_2\}}
        \operatorname{Res}_{s=p}\bigl(e^{s\eta}\widetilde\varphi_\zeta(s)\bigr).
    \end{multline}

    As $R\to\infty$ the decay $\widetilde\varphi_\zeta(s)=\Oh(1/s)$ kills the outer arc for $\eta>0$, and the Bromwich segment reconstructs Eq.~\eqref{eq:inv-laplace}; as $\delta\to0$ the keyhole $\mathcal C_\delta$ collapses onto the two faces of the cut\footnote{The small arcs of $\mathcal C_\delta$ encircling the endpoints $\pm i\omega_c$ contribute nothing in the limit: $\widetilde\varphi_\zeta$ stays bounded there while their length is $\Oh(\delta)$.}. Hence
    \begin{multline}\label{eq:decomp}
        \varphi_\zeta(\eta)=
        \sum_{p\in\{\pm i\nu_1,\pm i\nu_2\}}
        \operatorname{Res}_{s=p}\bigl(e^{s\eta}\widetilde\varphi_\zeta(s)\bigr)\\
        +\frac{1}{2\pi}\int_{-\omega_c}^{\omega_c}e^{i\omega\eta}\,h_\zeta(\omega)\,\dd\omega,
    \end{multline}
    where $h_\zeta(\omega)\equiv\widetilde\varphi_{\zeta,+}(i\omega)-\widetilde\varphi_{\zeta,-}(i\omega)$.

    The four poles are simple and conjugate, and $\widetilde\varphi_\zeta(\bar s)=\overline{\widetilde\varphi_\zeta(s)}$, so
    \begin{multline}\label{eq:osc}
        \sum_{p\in P}\operatorname{Res}_{s=p}\bigl(e^{s\eta}\widetilde\varphi_\zeta\bigr)\\
        =\sum_{k\in\{1,2\}}\bigl[\,C_k\cos(\nu_k\eta)+S_k\sin(\nu_k\eta)\,\bigr],
    \end{multline}
    for certain real constants $C_k, S_k$.

    For the cut integral, the definition of $D_\zeta$ in Eq.~\eqref{eq:dem-D} gives $s\,\widetilde\gamma_\zeta(s)=-{M_{\rm eff}}(s^2+\Omega_0^2)-D_\zeta(s)$, so Eq.~\eqref{eq:varphi(z)} splits as
    \begin{equation}\label{eq:split}
        \widetilde\varphi_\zeta(s)=\frac{\varphi_0}{s}
        -\frac{{M_{\rm eff}}\dot\varphi_0-{M_{\rm eff}}\Omega_0^2\varphi_0/s}{D_\zeta(s)} .
    \end{equation}

    For $0<\omega<\omega_c$, using $D_{\zeta,\pm}(i\omega)=A_\zeta(-\omega^2)\mp i\zeta j(\omega)$ from Eq.~\eqref{eq:F-lim},
    \begin{equation}\label{eq:jump}
        \frac{1}{D_{\zeta,+}(i\omega)}-\frac{1}{D_{\zeta,-}(i\omega)}
        =\frac{2i\,\zeta\,j(\omega)}{A_\zeta(-\omega^2)^2+\zeta^2 j(\omega)^2},
    \end{equation}
    and so,
    \begin{equation}\label{eq:h}
        h_\zeta(\omega)= \left(\dot\varphi_0+\frac{i\Omega_0^2\varphi_0}{\omega}\right)
        \frac{-2i\zeta{M_{\rm eff}}\,j(\omega)}{A_\zeta(-\omega^2)^2+\zeta^2 j(\omega)^2}.
    \end{equation}

    In any compact subinterval of $(0,\omega_c)$, $j$ is strictly positive, so $h_\zeta$ is continuous and thus integrable. Near $\omega=0$, $A_\zeta(-\omega^2)\to-{M_{\rm eff}}\Omega_0^2\neq0$ and $j(\omega)=\Oh(\omega^\beta)$, so $h_\zeta=\Oh(\omega^{\beta-1})$, integrable for $\beta>0$. Near $\omega=\omega_c$, $A_\zeta(-\omega^2)$ diverges logarithmically, and $h_\zeta=\Oh\bigl(|\ln(\omega_c-\omega)|^{-2}\bigr)$, also integrable. Given that $h_\zeta(-\omega)=\overline{h_\zeta(\omega)}$, $h_\zeta\in L^1\left(-\omega_c,\omega_c\right)$. Note that, since the cut integral in Eq.~\eqref{eq:decomp} is an inverse Fourier transform of a hermitian integrable function, it is real function and vanishes as $\eta\to\infty$.

    Collecting Eqs.~\eqref{eq:decomp}, \eqref{eq:osc}, and \eqref{eq:h}, the evolution of $\varphi_\zeta$ corresponds to Eq.~\eqref{eq:final-thm}, with
    \begin{equation}\label{eq:final}
         \varphi_{\zeta,\mathrm{tr}}(\eta) =
        \frac{1}{2\pi}\int_{-\omega_c}^{\omega_c}e^{i\omega\eta}\,h_\zeta(\omega)\,\dd\omega.
    \end{equation}
\end{proof}

\section{Power-law spectral densities with a sharp cutoff}

\begin{lemma}\label{lem:X}
Let $0<\varepsilon<e^{-1}$ and $\chi\in C_c^{\infty}([0,\varepsilon))$ with $\chi\equiv1$ near $0$. As $\eta\to\infty$:

\smallskip
(a) if $g\in C^{N}((0,\varepsilon))$ satisfies
\begin{equation}\label{partial-bounds}
\partial^k g(\omega)
=
\Oh \left(\omega^{\beta-k}|\ln \omega|^\mu\right),
\quad \omega\downarrow0,
\end{equation}
for $k=0,\ldots,N$, and some $\mu, \beta \in\mathbb{R}$ with $0<\beta+1<N$, then
\begin{equation}
\int_0^\varepsilon
\chi(\omega)g(\omega)e^{-i\eta \omega}\,\dd \omega
=
\Oh \left(\eta^{-\beta-1}\ln^\mu \eta\right);
\end{equation}

\smallskip
(b) for $\alpha>-1$ and every $N \in \mathbb{N}$,
\begin{multline}
\int_0^{\varepsilon}\!\chi(\omega)\,\omega^{\alpha}e^{-i\eta \omega}\,\dd\omega\\
=\Gamma(\alpha{+}1)\,e^{-i\pi(\alpha+1)/2}\eta^{-\alpha-1}
+\Oh(\eta^{-N});
\label{X:erd}
\end{multline}

\smallskip
(c) for fixed $a>0$,
\begin{multline}
L_a(\eta)\equiv
 \int_0^{\varepsilon}\frac{\chi(\omega)}{\ln^{2}(\omega/a)+\pi^{2}}\,
e^{-i\eta \omega} \dd \omega\\
=\frac{-i}{\eta\,\big[\ln^{2}(a\eta)+\pi^{2}\big]}
+\Oh\left(\eta^{-1}\ln^{-3}\eta\right).
\end{multline}
\end{lemma}

\begin{proof}[Proof of Lemma~\ref{lem:X}]
(a) Follows from splitting the integral at $\omega = \eta^{-1}$, the bounds of Eq.~\eqref{partial-bounds}, and the asymptotic equivalences
\begin{equation}\label{eq:first-endpoint}
\int_0^{1/\eta}\omega^\beta|\ln \omega|^\mu\dd \omega
\sim
\frac{1}{\beta+1}
\eta^{-\beta-1} \ln^\mu \eta,
\end{equation}
and
\begin{multline}\label{eq:second-endpoint}
\int_{1/\eta}^{\varepsilon}\omega^{\beta-N} |\ln \omega|^\mu\dd \omega
\\
\sim
\frac{1}{N-\beta-1}
\eta^{N-\beta-1}\ln^\mu \eta
\end{multline}
for $\eta \to \infty$, given by Karamata's integral theorem \cite[Prop.1.5.8]{Bingham:1989rv}.

(b) Follows from
\begin{multline}\label{eq:I0-note}
\int_0^{\infty}\chi(\omega)\,\omega^{\alpha}e^{-i\eta \omega}\,\dd\omega
=\int_0^{\infty}\omega^{\alpha}e^{-i\eta \omega}\,\dd\omega\\
+\int_0^{\infty}\big(\chi(\omega)-1\big)\,\omega^{\alpha}e^{-i\eta \omega}\,\dd \omega\\
=\Gamma(\alpha{+}1)\,e^{-i\pi(\alpha+1)/2}\eta^{-\alpha-1}+\Oh(\eta^{-N}),
\end{multline}
interpreted as Abel-regularized integrals; the remainder follows by repeated integration by parts.

(c) $\chi f$, $f\equiv[\ln^{2}(\omega/a)+\pi^{2}]^{-1}$, vanishes at both endpoints, so one integration by parts gives $L_a=-\tfrac{i}{\eta}\int_0^{\varepsilon}e^{-i\eta \omega}(\chi f)' \dd \omega$, with $f'=-2\ln(\omega/a)\big/\big\{\omega[\ln^{2}(\omega/a)+\pi^{2}]^{2}\big\}$. On $(0,\eta^{-1})$, $|e^{-i\eta \omega}-1|\leq\eta \omega$ gives $\int_0^{\eta^{-1}}e^{-i\eta \omega}f'\dd \omega=f(\eta^{-1})+\Oh(\ln^{-3}\eta)=[\ln^{2}(a\eta)+\pi^{2}]^{-1}+\Oh(\ln^{-3}\eta)$; on $(\eta^{-1},\varepsilon)$, one further integration by parts, using $f'(\eta^{-1})=\Oh(\eta\ln^{-3}\eta)$, gives $\int_{\eta^{-1}}^\varepsilon|(\chi f')'|\,\dd\omega=\Oh(\eta\ln^{-3}\eta)$.
\end{proof}

\begin{proof}[Proof of Theorem~\ref{thm:X}]
Since $h_\zeta(-\omega)=\overline{h_\zeta(\omega)}$, Eqs.~\eqref{eq:h} and~\eqref{eq:final} give
\begin{multline}\label{X:cut}
\varphi_{\zeta,\mathrm{tr}}(\eta)=\frac{1}{\pi}\int_{0}^{\omega_c}
\frac{2\zeta{M_{\rm eff}}\,\omega^p}{Q(\omega)}
\Big[\frac{\Omega_0^{2}\varphi_0}{\omega}\cos(\omega\eta)
\\[-2pt]
\quad+\dot\varphi_0\sin(\omega\eta)\Big]\,\dd\omega,
\end{multline}
where $Q(\omega) \equiv A_\zeta(-\omega^{2})^{2}+\zeta^{2}\omega^{2p}$, whose late-time behavior is governed by the endpoints $\omega=0$ and $\omega=\omega_c$, the interior contributing faster-decaying terms.

A partition of unity $\chi_0+\chi_m+\chi_c=1$ on $[0,\omega_c]$, with $\chi_0\equiv1$ near $0$, $\chi_c\equiv1$ near $\omega_c$, and $\chi_m\in C_c^{\infty}\big((0,\omega_c)\big)$, splits Eq.~\eqref{X:cut} as $\varphi_{\zeta,\mathrm{tr}}=I_0+I_m+I_c$. Since the integrand is smooth and $\operatorname{supp}\chi_m$ is compact in $(0,\omega_c)$, $I_m=\Oh(\eta^{-N})$ for every $N$.

\smallskip
\emph{Zero-frequency endpoint.} Noting that for $\omega\downarrow 0$, $A_\zeta(-\omega^2)\to-{M_{\rm eff}}\Omega_0^2\neq0$, we define $r\equiv Q^{-1}-1/({M_{\rm eff}}^2\Omega_0^{4})$ and write the zero-frequency contribution as
\begin{equation}
I_0(\eta)=I_0^{(0)}(\eta)+I_0^{(r)}(\eta),
\end{equation}
where $I_0^{(0)}$ is obtained by replacing $Q^{-1}$ with $1/({M_{\rm eff}}^2\Omega_0^4)$. Lemma~\ref{lem:X}(b), applied at $\alpha=p-1, p$, gives
\begin{align}
I_0^{(0)}(\eta)=&{\frac{2\zeta\,\Gamma(p)\cos(\tfrac{\pi p}{2})}
{\pi{M_{\rm eff}}\Omega_0^{2}}}\,\varphi_0\,\eta^{-p}
\nonumber\\
&+{\frac{2\zeta\,\Gamma(p{+}1)\cos(\tfrac{\pi p}{2})}
{\pi{M_{\rm eff}}\Omega_0^{4}}}\,\dot\varphi_0\,\eta^{-p-1}
+\Oh(\eta^{-N}).
\end{align}

It remains to estimate the contribution containing $r$, $I_0^{(r)}(\eta)$. In the notation of Eq.~\eqref{eq:A}, the substitution $u=\nu^{2}$ gives
\begin{equation}
A_\zeta(-\omega^{2})=-{M_{\rm eff}}(\Omega_0^{2}-\omega^{2})+\frac{2\zeta}{\pi}\,\Delta(\omega),
\label{X:A}
\end{equation}
\begin{equation}
\Delta(\omega)\equiv\omega^{2}\,\mathrm{p.v.}\!\int_0^{\omega_c}
\frac{\nu^{p-1}}{\nu^{2}-\omega^{2}}\,\dd\nu.
\label{X:Delta}
\end{equation}
With $q\equiv\min\{p,2\}$,
\begin{equation}
|\partial_\omega^{k}\Delta|\leq
C_k\,\omega^{q-k}\big[1+|\ln(\omega)|\big],
\qquad 0<\omega<\varepsilon.
\label{X:sym}
\end{equation}
Indeed, for $0<p<2$ the substitution $\nu=\omega u$ gives $\Delta=h_p\omega^{p}+\omega^{2}B(\omega^{2})$, with $h_p\equiv\mathrm{p.v.}\!\int_0^{\infty}u^{p-1}(u^{2}-1)^{-1}\dd u$ and $B(z)\equiv-\int_{\omega_c}^{\infty}\nu^{p-1}(\nu^{2}-z)^{-1}\dd\nu$ analytic near $z=0$; for $p=2$, explicitly $\Delta=\tfrac{\omega^{2}}{2}\ln[(\omega_c^{2}-\omega^{2})/\omega^{2}]$; for $p>2$ the identity $\nu^{p-1}=\nu^{p-3}(\nu^{2}-\omega^{2})+\omega^{2}\nu^{p-3}$ yields $\Delta_p=\tfrac{\omega_c^{p-2}}{p-2}\,\omega^{2}+\omega^{2}\Delta_{p-2}$, iterating down to the previous cases.

Since {$Q\geq{M_{\rm eff}}^2\Omega_0^{4}/2$} near $0$, the derivative bound of Eq.~\eqref{X:sym} transfers to $r$. The two $r$-dependent amplitudes are $\chi_0\omega^{p-1}r$ and $\chi_0\omega^{p}r$. Lemma~\ref{lem:X}(a) therefore gives the bounds $\Oh(\eta^{-p-q}\ln\eta)$ and $\Oh(\eta^{-p-q-1}\ln\eta)$, respectively. Consequently,
\begin{multline}\label{X:I0}
I_0^{(r)}(\eta) = \varphi_0\,\Oh\big(\eta^{-p-q}\ln\eta\big)\\
+\dot\varphi_0\,\Oh\big(\eta^{-p-q-1}\ln\eta\big).
\end{multline}
Define $s(\eta)\equiv I_0(\eta)+I_m(\eta)$. Since $I_m=\Oh(\eta^{-N})$ for every $N$, Eq.~\eqref{X:I0} proves Eq.~\eqref{X:s}.

\smallskip
\emph{Sharp-cutoff endpoint.} We now analyze the contribution localized near $\omega=\omega_c$. Subtracting the pole of the principal value,
\begin{align}
\Delta(\omega)=&\int_0^{\omega_c}\frac{1}{\nu-\omega}
\Big[\frac{\omega^{2}\nu^{p-1}}{\nu+\omega}
-\frac{\omega^{p}}{2}\Big]\dd\nu
\nonumber\\
&+\frac{\omega^{p}}{2}\,
\ln\Big(\frac{\omega_c-\omega}{\omega}\Big),
\label{X:pole}
\end{align}
where the bracket vanishes at $\nu=\omega$, so the first term is smooth. Hence, by Eq.~\eqref{X:A}, we may express
\begin{equation}
A_\zeta(-\omega^{2})=\frac{\zeta}{\pi}\,\omega^{p}
\ln\Big(\frac{\omega_c-\omega}{\omega_c}\Big)+\beta(\omega),
\label{X:Aedge}
\end{equation}
with $\beta$ smooth near $\omega_c$; in particular the limit in Eq.~\eqref{X:betac} exists, with $\beta_c=\beta(\omega_c)$. Introducing the variable $x\equiv(\omega_c-\omega)/\omega_c$, $A_\zeta = c(x)\ln x+b(x)$, with $c(x)\equiv\tfrac{\zeta}{\pi}\omega_c^{p}(1-x)^{p}$ and $b \equiv \beta\bigl(\omega_c(1-x)\bigr)$. Additionally, defining
\begin{equation}\label{X:shift}
\ln x^* \equiv-\frac{b(0)}{c(0)}, \quad 
\ell \equiv\ln\frac{x}{x^*},
\quad
b_{\mathrm r} \equiv b-\frac{b(0)}{c(0)}\,c,
\end{equation}
we reexpress $A_\zeta=c\,\ell+b_{\mathrm r}$, with $b_{\mathrm r}(0)=0$. By definition $\pi^{2}c(x)^{2} = \zeta^{2}\omega^{2p}$, so 
\begin{align}
Q&=c^{2}\,\bigl[\ell^{2}+\pi^{2}\bigr]\,D,
\nonumber\\ 
D&\equiv1+\frac{2b_{\mathrm r}}{c}\,\frac{\ell}{\ell^{2}+\pi^{2}}
+\frac{b_{\mathrm r}^{2}}{c^{2}}\,\frac{1}{\ell^{2}+\pi^{2}}.
\label{X:D}
\end{align}

Note that, since $D\to1$ as $x\downarrow0$, after reducing $\varepsilon$ if necessary we may assume $D\geq\tfrac12$ on $(0,\varepsilon)$. Choosing $\operatorname{supp}\chi_c \subset [\omega_c(1-\varepsilon),\omega_c]$ and defining $W\equiv-2\zeta{M_{\rm eff}}\,\omega^{p}/Q$,
\begin{multline}\label{Ic}
    I_c(\eta)
    ={}
    -\frac{\omega_c}{\pi}
    \int_0^\varepsilon
    \chi_{c}\bigl(\omega_c(1-x)\bigr)
    W\bigl(\omega_c(1-x)\bigr)
    \\
    \times
    \Big[
    \frac{\Omega_0^2\varphi_0}{\omega_c(1-x)}
    \cos\bigl(\eta \, \omega_c (1-x)\bigr)\\
    +
    \dot\varphi_0
    \sin\bigl(\eta \, \omega_c (1-x)\bigr)
    \Big]\dd x.
\end{multline}

$W=-2\pi^{2}{M_{\rm eff}}\,(1-x)^{-p}\big/\bigl\{\zeta\omega_c^{p}\,[\ell^{2}+\pi^{2}]\,D\bigr\}$ splits as
\begin{align}
W\big(\omega_c(1-x)\big)
&=-{\frac{2\pi^{2}{M_{\rm eff}}}{\zeta\omega_c^{p}}}\,\frac{1}{\ell^{2}+\pi^{2}}+\rho(x),
\label{X:W}\\
\rho(x)&\equiv-{\frac{2\pi^{2}{M_{\rm eff}}}{\zeta\omega_c^{p}}}\,
\frac{(1-x)^{-p}D^{-1}-1}{\ell^{2}+\pi^{2}},\nonumber
\end{align}
so, substituting Eq.~\eqref{X:W} into Eq.~\eqref{Ic}, we write the sharp-cutoff contribution as
\begin{equation}\label{Ic-decomp}
    I_c(\eta) = I_c^{(0)}(\eta)+I_c^{(\rho)}(\eta),
\end{equation}
where $I_c^{(0)}$ is obtained by replacing $W$ with {$-\frac{2\pi^2{M_{\rm eff}}}{\zeta\omega_c^p[\ell^2+\pi^2]}$}.

To evaluate $I_c^{(0)}$, we decompose
\begin{equation}
    \frac{1}{1-x} = 1+\frac{x}{1-x}.
\end{equation}

The integrand of $I_c^{(0)}$ containing $x/(1-x)$ has an amplitude equal to a smooth function times $x\,[\ell^{2}+\pi^{2}]^{-1}$, which by Lemma~\ref{lem:X}(a), with $(\beta,\mu)=(1,-2)$, contributes as
$\Oh\left( \eta^{-2}\ln^{-2}(\eta) \right)$.

Consequently,
\begin{equation}\label{Ic0-1}
    \begin{aligned}
I_c^{(0)}(\eta)
={}&
{\frac{2\pi{M_{\rm eff}}\omega_c^{1-p}}{\zeta}}
\bigg[
\frac{\Omega_0^2\varphi_0}{\omega_c}
\Re\left(
e^{i\omega_c \eta}L_{x^*}(\omega_c \eta)
\right)
\\
&
\hspace{5em}
+
\dot\varphi_0
\Im\left(
e^{i\omega_c \eta}L_{x^*}(\omega_c \eta)
\right)
\bigg]
\\
&\hspace{5em}
+ \Oh\left( \eta^{-2}\ln^{-2}(\eta) \right),
\end{aligned}
\end{equation}
where $L_a$ is defined in Lemma \ref{lem:X}(c), here with $a=x^*$. Since $\ln(x^*\omega_c\eta)=\ln(\eta/\eta_*)$, the Lemma gives
\begin{equation}\label{I0c-rel}
\begin{split}
\Re\left(e^{i\omega_c \eta}L_{x^*}(\omega_c \eta)\right)
=&
\frac{\sin(\omega_c \eta)}{\omega_c \eta\,\bigl[\ln^{2}(\eta/\eta_*)+\pi^{2}\bigr]}\\
& \qquad \qquad+
\Oh\left(\eta^{-1}\ln^{-3} \eta \right),
\\
\Im\left(e^{i\omega_c \eta}L_{x^*}(\omega_c \eta)\right)
=&
-\frac{\cos(\omega_c \eta)}{\omega_c \eta\,\bigl[\ln^{2}(\eta/\eta_*)+\pi^{2}\bigr]}\\
& \qquad \qquad +
\Oh\left( \eta^{-1}\ln^{-3}\eta \right).
\end{split}
\end{equation}

Substituting the relations of Eq.~\eqref{I0c-rel} into Eq.~\eqref{Ic0-1} we obtain
\begin{equation}\label{X:Ic0}
    \begin{aligned}
        I_c^{(0)}(\eta)
        ={}&
        -{\frac{2\pi{M_{\rm eff}}}{\zeta\omega_c^{p}}}\,
        \frac{
        \dot\varphi_0\cos(\omega_c\eta)
        -
        \dfrac{\Omega_0^{2}\varphi_0}{\omega_c}
        \sin(\omega_c\eta)
        }{
        \eta\,\bigl[\ln^{2}(\eta/\eta_*)+\pi^{2}\bigr]
        }
        \\
        &+
        \Oh\left(
        \eta^{-1}\ln^{-3}(\eta)
        \right).
    \end{aligned}
\end{equation}

To evaluate $I_c^{(\rho)}$, note that the factor $(1-x)^{-p}$ is smooth and equals $1$ at $x=0$, so $(1-x)^{-p}D^{-1}-1$ obeys the same bounds as $D^{-1}-1$ up to terms carrying an extra factor of $x$. By Eq.~\eqref{X:D}, $D-1$ is a sum of functions smooth near $x=0$ and vanishing at $x=0$, weighted by $\ell/(\ell^{2}+\pi^{2})$ and $(\ell^{2}+\pi^{2})^{-1}$. Since $|\partial_x^{k}[\ell/(\ell^{2}+\pi^{2})]|=\Oh\bigl(x^{-k}(|\ln x|)^{-1}\bigr)$, $|\partial_x^{k}(\ell^{2}+\pi^{2})^{-1}|=\Oh\bigl(x^{-k}(|\ln x|)^{-2}\bigr)$, and $D\geq\tfrac12$, Leibniz' rule gives $|\partial_x^{k}(D^{-1}-1)|\leq C_k\,x^{-k}(|\ln x|)^{-1}$ near $x=0$. Therefore, for  $\varepsilon<1$,
\begin{equation}\label{X:rho}
|\partial_x^{k}\rho(x)|\leq C_k\,x^{-k}\,(|\ln x|)^{-3}, \quad x\in (0,\varepsilon).
\end{equation}

Since $(1-x)^{-1}$ is smooth on $[0,\varepsilon]$, multiplication by this factor preserves the derivative bounds in Eq.~\eqref{X:rho}, so Lemma~\ref{lem:X}(a), applied with $(\beta,\mu)=(0,-3)$, yields
\begin{equation}\label{Icrho}
    I_c^{(\rho)}(\eta) =
    \Oh\left(\eta^{-1}\ln^{-3}(\eta)\right).
\end{equation}

Equations \eqref{Ic-decomp}, \eqref{X:Ic0} and \eqref{Icrho}, and the harmonic identity $\dot\varphi_0\cos y-\dfrac{\Omega_0^{2}\varphi_0}{\omega_c}\sin y=R_0\cos(y+\delta)$ imply
\begin{equation}\label{Ic-final}
    I_c=a_c(\eta)\cos(\omega_c\eta+\delta)+\Oh(\eta^{-1}\ln^{-3}(\eta)).
\end{equation}

Together with the zero-frequency endpoint contribution, this proves Eq.~\eqref{X:main}.

\end{proof}